\documentclass[12pt,letter]{article}
\usepackage[latin1]{inputenc}
\usepackage{xspace}
\usepackage{amsmath,amssymb,amsfonts,epic}
\usepackage[dvips]{graphicx}
\usepackage{graphics,epsfig}
\usepackage{fancybox}
\usepackage{amsbsy}
\usepackage{amscd}
\usepackage{amsgen}
\usepackage{amsopn}
\usepackage{amstext}
\usepackage{ifthen}
\usepackage{amsxtra}
\usepackage{color}
\usepackage{fancyhdr}
\usepackage{fullpage}
\usepackage{float}
\usepackage[round]{natbib}
\definecolor{hellgrau}{gray}{0.4}

\newcommand{\e}{\varepsilon}
\newcommand{\N}{\mathbb{N}}

\newcommand{\A}{\mathcal{A}}

\newcommand{\pr}{\mathbb{P}_\theta}

\newcommand{\esp}{\mathbb{E}_\theta}
\newcommand{\pro}{\mathbb{P}_{0}}
\newcommand{\espo}{\mathbb{E}_{0}}
\newcommand{\Q}{Q_\theta}

\newcommand{\Qo}{Q_{\theta_0}}
\newcommand{\T}{\mathcal{T}}                         

\newcommand{\X}{\mathbb{X}}

\newcommand{\1}[1]{1\! \mathrm{l}\{ #1\}}

\newcommand{\E}{\mathcal{E}}

\newtheorem{thm}{Theorem}

\newtheorem{defi}{Definition}

\newtheorem{assumption}{Assumption}
\newtheorem{remark}{Remark}
\newtheorem{lemma}{Lemma}
\newenvironment{proof}[1][Proof] {{\bf{#1}}}{\hfill $\square$}

\newcommand{\smalleq}{\fontsize{14.4}{18}\selectfont} 

\begin{document}
\title{Parameter Estimation in multiple-hidden i.i.d. models from biological multiple alignment}
\author{Ana~Arribas-Gil
\thanks{Departamento de Estadística. Universidad Carlos III de Madrid. C/ Madrid 126, 28903 Getafe, Spain. E-mail: aarribas@est-econ.uc3m.es}}
\lhead[\fancyplain{}{ \thepage}]%
      {\fancyplain{}{\emph A. Arribas-Gil}}
\rhead[\fancyplain{}{\emph Parameter estimation in multiple-hidden i.i.d. models}]%
      {\fancyplain{}{ \thepage}}
\maketitle

\begin{abstract}
In this work we deal with parameter estimation
in a latent variable model, namely the multiple-hidden i.i.d. model, which is derived from multiple alignment algorithms.
We first provide a rigorous formalism for the homology structure of $k$ sequences related
by a star-shaped phylogenetic tree in the context of multiple alignment based on indel evolution models. 
We discuss  possible definitions of likelihoods and compare them to the criterion used in multiple alignment algorithms.
Existence of two different Information divergence rates is
established and a divergence property is shown under additional
assumptions. This would yield consistency for the parameter in
parametrization schemes for which the divergence property holds.
We finally extend the definition of the multiple-hidden i.i.d. model and the results obtained to the case in which the sequences are related by an arbitrary phylogenetic tree. Simulations illustrate different cases which are not
covered by our results.
\end{abstract}


\section{Introduction}
Biological sequence alignment is one of the fundamental tasks in bioinformatics. Sequences are aligned to identify regions of similarities that can be used to determine structural and functional motifs in a sequence, to infer gene functions or to derive evolutionary relationships between sequences. Aligning two sequences, which are supposed to descend from a common ancestor, consists in retrieving the places where substitutions, insertions and deletions have occurred during evolution. The first alignment methods, namely scored-based methods, used dynamic programming algorithms with fixed score parameters to find an optimal alignment (see Durbin {\em et al.}, 1998, for an overview). But since an alignment aims at reconstructing the evolution history of the sequences, choosing these score parameters in the most objective way to have an evolutionary meaning seems to be an important issue.
\citet{TKF1} proposed the first rigorous model of sequence evolution including {\em indels} (insertions and deletions), referred to as the TKF91 model. Based on this model, they were the first to provide a maximum likelihood approach to jointly estimate the alignment of a pair of DNA sequences and the evolution parameters. The alignment problem in this context fits into the pair hidden Markov model (pair-HMM), as first described in \citet{Durbin}, ensuring the existence of efficient algorithms based on dynamic programming methods to compute the likelihood of two sequences and retrieve an alignment. That is one of the reasons why TKF91 based alignment methods have become popular. Indeed, they have been further developed in \citet{HeinWiuf}, \citet{Metzler1}, \citet{Metzler2} and \citet{Miklos} among others, and this despite the lack of theoretical support for the estimation procedures in this framework during years. \citet{Argamat} were the first to study the statistical properties of parameter estimation procedures in pair-HMMs.

In the last years these methods have also been extended to the case of multiple alignment. In this context we deal with more than two sequences and we have to take into account the evolutionary relationships between the sequences, which are represented by a phylogenetic tree. Multiple alignment methods applying the TKF91 model on a tree are for instance those of \citet{Steel}, \citet{HolmesBruno}, \citet{Hein03} and \citet{Hein4}. They generalize pair-HMMs to more complex hidden variable models and propose maximum likelihood or Bayesian approaches for the joint estimation of evolution parameters and multiple alignments given a phylogenetic tree. However, since both alignment and phylogenetic tree aims at reconstructing the evolutionary history of the sequences, estimating the alignment from a fixed phylogenetic tree may biased the result. The ideal procedure would consist in jointly estimating alignments and phylogenetic trees from a set of unaligned sequences. This problem has been recently tackled, in the context of indel evolution models, by \citet{Metzler3}, \citet{Hein5} and \citet{StatAlign}. However, as it was the case during years for the pair-HMMs, no theoretical support is provided for the estimation procedures in any of these contexts.

This work is concerned with the study of statistical properties of parameter estimation in latent variable models derived from multiple alignment algorithms where the phylogenetic tree relating the observed sequences is supposed to be known. The paper is organized as follows.

In Section 2, we motivate the problem, discuss some models of sequence evolution and describe the homology structure in the context of multiple alignment of a set of sequences related by a star-shaped phylogenetic tree and evolving under the TKF91 model of sequence evolution.

In Section 3 we present the multiple-hidden i.i.d. model on a star tree. We discuss possible definitions of likelihoods and compare them with the criterion which is actually considered in multiple alignment algorithms. We analyze the case in which only two sequences are considered to show that our model is consistent with the pair-HMM.

In Section 4, we investigate asymptotic properties of estimators under the hidden i.i.d. model for the definitions of likelihoods that we have considered.
We first prove the existence of {\em Information divergence rates}, which are the difference between the limiting values of the log-likelihoods at the (unknown) true parameter and at another parameter value. We then prove that they are uniquely minimized at the true value of the parameter (divergence property) for some parametrization schemes. Following classical arguments, this would yield consistency for the parameter in those cases in which the divergence property holds.

In Section 5 we extend the definitions of the multiple-hidden i.i.d. model and the results obtained to the general case in which the sequences are related by an arbitrary phylogenetic tree.

Finally, in Section 6, we illustrate via some simulations the behavior of the divergence rates in different cases in which the divergence property is not established. The paper ends with a discussion on this work.

\section{Motivation: models of sequence evolution and the homology structure}

In the multiple alignment problem the observations consist in $k$ ($k>2$)
sequences $X^1_{1:n_1},..., X^k_{1:n_k}$, where $n_i$ is the length of sequence $i$ and
$X^i_{1:n_i}=X^i_1\dots X^i_{n_i}$, with values in a finite
alphabet ${\cal A}$ (for instance $\A=\{A,C,G,T\}$ for DNA sequences). It is
assumed that the sequences are related by a phylogenetic tree,
that is, a tree where the nodes represent the sequences and the
edges represent the evolutionary relationships between them. The
observed sequences are placed at the $k$ leaves of the tree,
whereas the inner nodes stand for ancestral (non-observable)
sequences. The most ancestral sequence is placed at the root,
${\cal R}$, of the tree. The choice of the root assigns to each
edge a direction (from the root to the leaves) and to each inner
node its descendants nodes, but since the evolutionary process
between the sequences is usually assumed to be time reversible,
the placement of the root node is irrelevant (cf. Thatte, 2006).
A path from the root to a leaf represents the evolution through
time and through a series of intermediate sequences of the
ancestral sequence, leading to the corresponding observed
sequence. The evolution on each edge (from its \emph{parent} node
to its \emph{child} node) is described by some evolution
process. We assume that the same evolution process works on
every edge of the tree. A main hypothesis is that the evolution
processes working on two edges with the same \emph{parent} node
are independent,
i.e. a sequence evolves independently to each one of its descendants.

\subsection{Models of sequence evolution}

Mutations in a sequence during the evolution process can be produced by many different factors. However, there are two evolutionary events that play a major role: substitutions of a nucleotide by a different one in a given position of a sequence, and insertions or deletions of single positions or sequence fragments.

The process of substitutions has been studied in depth during years, and is usually taken to be a continuous time Markov chain on the state space of nucleotides (Felsenstein, 2004; Tavaré, 1986). The process of insertions and deletions has not received the same attention and there is more place for discussion. \citet{TKF1} proposed in a pioneering paper the first indel evolution model, and since then many variants have been considered. The importance of this model is that it makes the alignment fit into the concept of pair-HMM, as we have already mentioned.

In the pair-HMM for pairwise sequence alignment the indel process and the substitution process are combined to model the whole evolution process.
Indeed, the hidden Markov chain corresponds to what we usually call the \emph{bare} alignment, that is, an alignment without specification of the particular nucleotides at each position of the sequences. Conditionally on a realization of this hidden process, the observed sequences are emitted according to the substitution model (see Durbin {\em et al.}, 1998, and Arribas-Gil {\em et al.}, 2006, for details).

So, in the pair-HMM the indel evolution process characterizes the hidden stochastic process of the alignment, whereas the substitution process corresponds to the emission functions of the observed sequences. As we will see, that is also the case for the multiple alignment model that we study in this paper.
Since the asymptotic properties of estimators in such a model are more related to the structure of the hidden process than to the emission functions, which can take a general form (see Arribas-Gil {\em et al.}, 2006), we will focus our attention on the indel process.

\subsubsection{The TKF91 model}
Let us briefly recall how the TKF91 model works on pairwise alignments. This model is
formulated in terms of \emph{links} and associated letters. To
each \emph{link} is associated a letter that undergoes changes,
independently of other letters, according to a reversible
substitution process. The insertion and deletion process is
described by a birth-death process on these \emph{links}. Indeed, a \emph{link}
and its associated letter is deleted at the rate $\mu>0$. While a
\emph{link} is present it gives rise to new \emph{links} at the rate $\lambda$.
A new \emph{link} is placed immediately to the right of the \emph{link} from
which it originated, and the associated letter is chosen from the
stationary distribution of the substitution process. At the very
left of the sequence is a so-called immortal \emph{link} that never dies
and gives rise to new \emph{links} at the rate $\lambda$. We need the
death rate per \emph{link} to exceed the birth rate per \emph{link} to have a
distribution of sequence lengths. Indeed, if $\lambda < \mu$ then
the equilibrium distribution of length sequence is geometric with
parameter $\lambda / \mu$.

Let $p^H_n(t)$ be the probability that a normal \emph{link} survives and
has $n$ descendants, including itself, after a time $t$. Let
$p^N_n(t)$ be the probability that a normal \emph{link} dies but leaves $n$
descendants after a time $t$. Finally let $p^I_n(t)$ be the
probability that an immortal \emph{link} has n descendants, including
itself, after a time $t$. Here $H$ stands for homologous, $N$ for
non-homologous and $I$ for immortal. We have:
\begin{equation}\label{TKFprocess}
\begin{array}{llll}
p^H_n(t) &= & e^{-\mu t}[1-\lambda
\beta(t)][\lambda \beta(t)]^{n-1} & \mbox{for } \,\,
 n\geq 1 \\
p^N_n(t) &= & \mu \beta(t) & \mbox{for }\,\,n=0 \\
&=& [1-e^{-\mu t}-\mu \beta(t)][1-\lambda \beta(t)][\lambda
\beta(t)]^{n-1} & \mbox{for } \,\,n\geq 1\\
 p^I_n(t) &= & [1-\lambda \beta(t)][\lambda \beta(t)]^{n-1} &
\mbox{for } \,\, n\geq 1
\end{array}
\end{equation}
where $$\beta(t)=\frac{1-e^{(\lambda-\mu)t}}{\mu-\lambda
e^{(\lambda-\mu )t}}.$$
Conceptually, $e^{-\mu t}$ is the
probability of ancestral residue survival, $\lambda \beta(t)$ is
the probability of more insertions given one or more existent
descendants and $\kappa(t):=\frac{1-e^{-\mu t}-\mu
\beta(t)}{1-e^{-\mu t}}$ is the probability of insertion given
that the ancestral residue did not survive. See \citet{TKF1} for
details.

If we want to investigate the asymptotic properties of parameter
estimators we must consider observed sequences of growing lengths.
However, this is not possible under the hypothesis of the TKF91
model. Indeed, the ancestral sequence length distribution depends
on $\lambda/\mu$, and so, for a given value of these parameters we
can not make the ancestral sequence length to tend to infinity. As
one would expect (and as we will show later) the lengths of the
observed sequences are equivalent to the length of the root
sequence, so under this setup we can not expect to observe
infinitely long sequences.

Following the ideas in Metzler (2003), we will consider the case in which the TKF91 model can
produce long sequences, that is, the case where $\lambda=\mu$.
With this configuration, finite length sequences are to be
considered as cut out of very much longer sequences between known
homologous positions. The length of the ancestral sequence is now
considered to be non random.

We will note $q^H_n(t)$ and $q^N_n(t)$ the probability
distributions of the number of descendants for a normal
\emph{link} under these assumptions. We do not need to consider
the distribution for the immortal \emph{link} anymore, since now
all the positions on the observed sequences are descendants of
normal \emph{links}.\\
Since $\lim_{\mu \rightarrow \lambda}\beta(t)=\frac{t}{1+\lambda
t}$ we get
\begin{eqnarray}\label{newprocess}
q^H_n(t) =\lim_{\mu \rightarrow \lambda} p^H_n(t) & =& e^{-\lambda
t}\frac{1}{1+\lambda t}\left(\frac{\lambda t}{1+\lambda
t}\right)^{n-1}\hspace{2.3cm}  \mbox{for } \,\,
 n\geq 1 \nonumber\\
q^N_n(t) =\lim_{\mu \rightarrow \lambda} p^N_n(t) &= & \frac{\lambda t}{1+\lambda t}
\hspace{5.4cm} \mbox{for }\,\,n=0
\\
                                                 &= & \left(\frac{1}{1+\lambda t}-e^{-\lambda t}\right)\frac{1}{1+\lambda t}
\left(\frac{\lambda t}{1+\lambda t}\right)^{n-1} \,\, \mbox{for } \,\,n\geq 1 \nonumber
\end{eqnarray}

The main drawback of the TKF91 model is that insertions and deletions can only be produced at one nucleotide at a time. More realistic indel evolution models based on the TKF91 model are, for instance, those of \citet{TKF2}, \citet{Miklos} or \citet{Ana_Metzler}. For the sake of simplicity, in this work we will just consider the TKF91 indel model. However, the homology structure and the multiple-hidden i.i.d. model presented here can be extended to the case in which other indel models are considered.

\subsection{A star tree}

Let us now consider a $k$-star phylogenetic tree, that is, a tree
with a root, $k$ leaves and no inner nodes. See Figure \ref{star}
for an example. We will note $t_i$, $i=1,\dots,k$, the branches
lengths, that is the evolutionary time separating each sequence to
the root. In this context, an alignment of the $k$ sequences and
the root consists in a composition of the $k$ pairwise alignments
of the root with any of the observed sequences. This is done as
follows. Two characters $X^i_j$ and $X^l_h$ will be aligned in the
same column if and only if they are homologous to the same
character of the root sequence. So there is a column for each
nucleotide at the root containing all its homologous positions on
the leaves, and between two columns of this kind, there is one
column for each inserted position on the leaves between the two
corresponding nucleotide positions at the root. Insertions to the
root sequence occur independently on each sequence and we assume
that the probability of having two insertions on different
sequences at the same time is 0. That is why insertion columns are
composed by one nucleotide position in some of the sequences and
gaps in all the others.
\begin{figure}
\begin{center}
\begin{minipage}{7cm}
\setlength{\unitlength}{1mm}
\begin{picture}(60,35)
\thicklines
\put(35,30){\line(-4,-3){31}}
\put(35,30){\line(-2,-3){16}}
\put(35,30){\line(-1,-3){8}}
\put(35,30){\line(1,-3){8}}
\put(35,30){\line(2,-3){16}}
\put(35,30){\line(4,-3){31}}
\put(34,32){${\cal R}$: \texttt{ACCT}}
\put(9,13){$t_1$}
\put(0,2){$X^1$} \put(-4.5,-2){\texttt{ACCGGT}}
\put(19,13){$t_2$}
\put(15,2){$X^2$} \put(13,-2){\texttt{ACT}}
\put(26,13){$t_3$}
\put(25,2){$X^3$} \put(23,-2){\texttt{ACT}}
\put(42,13){$t_4$}
\put(41,2){$X^4$} \put(34,-2){\texttt{GCCAT}}
\put(48,13){$t_5$}
\put(50,2){$X^5$} \put(49.5,-2){\texttt{CCT}}
\put(58,13){$t_6$}
\put(64,2){$X^6$} \put(61,-2){\texttt{ACCT}}
\end{picture}
 \end{minipage}
\hspace{0.4cm}\begin{minipage}{5.5cm}
{\scriptsize \textbf{Pairwise alignments:\vspace{0.2cm}\\ }
\begin{tabular}{rlrlrl}
\hspace{-0.2cm}${\cal R}$:&\texttt{ACC\textcolor[gray]{0.5}{--}T}&${\cal R}$:&\texttt{ACCT}&${\cal R}$:&\texttt{ACCT}\\
\hspace{-0.2cm}$X^1$:&     \texttt{ACC\textcolor[gray]{0.5}{GG}T}&$X^2$:&     \texttt{A-CT}&$X^3$:&     \texttt{AC-T}\\
\vspace{-0.2cm}\\
\hspace{-0.2cm}${\cal R}$:&\texttt{ACC\textcolor[gray]{0.5}{-}T}&${\cal R}$:&\texttt{ACCT}&${\cal R}$:&\texttt{ACCT}\\
\hspace{-0.2cm}$X^4$:&     \texttt{GCC\textcolor[gray]{0.5}{A}T}&$X^5$:&     \texttt{-CCT}&$X^6$:&     \texttt{ACCT}\\
\end{tabular}}\vspace{0.2cm}\\
{\scriptsize \begin{tabular}{l}\textbf{Multiple}\\ \textbf{alignment:} \end{tabular}
\begin{tabular}{rl}
$X^1$:&      \texttt{ACC\textcolor[gray]{0.5}{GG-}T}\\
$X^2$:&      \texttt{A-C\textcolor[gray]{0.5}{---}T}\\
$X^3$:&      \texttt{AC-\textcolor[gray]{0.5}{---}T}\\
$X^4$:&      \texttt{GCC\textcolor[gray]{0.5}{--A}T}\\
$X^5$:&      \texttt{-CC\textcolor[gray]{0.5}{---}T}\\
$X^6$:&      \texttt{ACC\textcolor[gray]{0.5}{---}T}
\end{tabular}}
 \end{minipage}
\end{center}
\caption{\emph{A 6-star phylogenetic tree and example of multiple alignment. ${\cal R}$ stands for the ancestral sequence  and $X^i$ stand for observed sequences. Letters in dark represent descendants of a position in the ancestral sequence whereas letters in gray represent insertions.}}\label{star}
\end{figure}
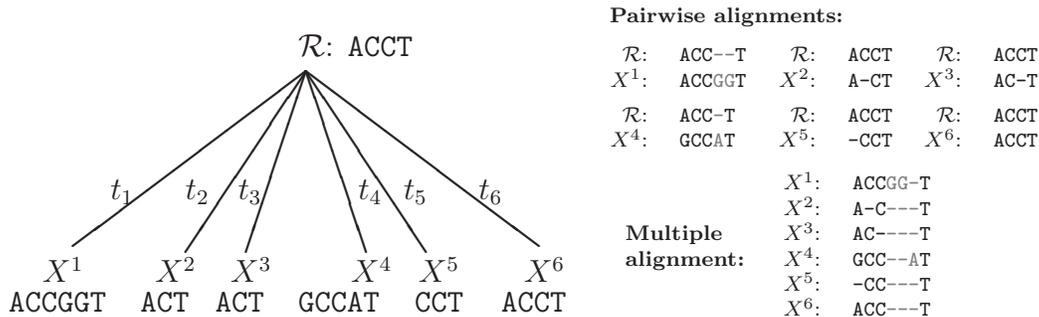
We know that under the TKF91 indel model the pairwise alignment is a Markov chain on the state space $\{\stackrel{\texttt{B}}
{\texttt{{\scriptsize B}}},\stackrel{\_}{\stackrel{\tiny
\phantom{.}}{\texttt{{\scriptsize B}}}},
\stackrel{\texttt{{\scriptsize B}}}{\stackrel{\texttt{{\tiny
\phantom{.}}}}{\_}} \}$ (see Metzler {\em et al.}, 2001).
Let us precise that from now on the word alignment will denote indistinctly the whole alignment, that is the reconstruction of the whole evolution process, including substitutions, of a set of sequences, or, as in this case, the \emph{bare} alignment, that is, the reconstruction of the indel process only. We recall that when we model the alignment of a set of sequences as a hidden variable model, the \emph{bare} alignment is which corresponds to the hidden process.

In contrast to the pairwise alignment case, when we apply the TKF91 indel evolution model to multiple
alignment we do not get a Markov chain on the set of all possible
multiple alignment columns. In fact, Markov models for multiple
alignment exist but states do not exactly correspond to alignment
columns. Indeed, insertion states in these models describe not
only an insertion on one sequence but also a kind of ``memory" of
what is happening in other sequences (see Holmes and Bruno, 2001, and Hein {\em et al.}, 2003, for
instance). This is because the Markov dependence for pairwise
root-leaf alignments applies independently on each sequence due to
the branch independence of the evolution process. So that, an
alignment column describing an insertion on sequence $i$ depends
on the last column of the alignment describing any evolutionary
event on sequence $i$, but there may be several
alignment columns describing insertions on other sequences between these two columns. See Figure~\ref{mapa} for an illustration.\\
\begin{figure}
\vspace{-0.35cm}
\centering
\includegraphics[width=9cm,height=7.4cm]{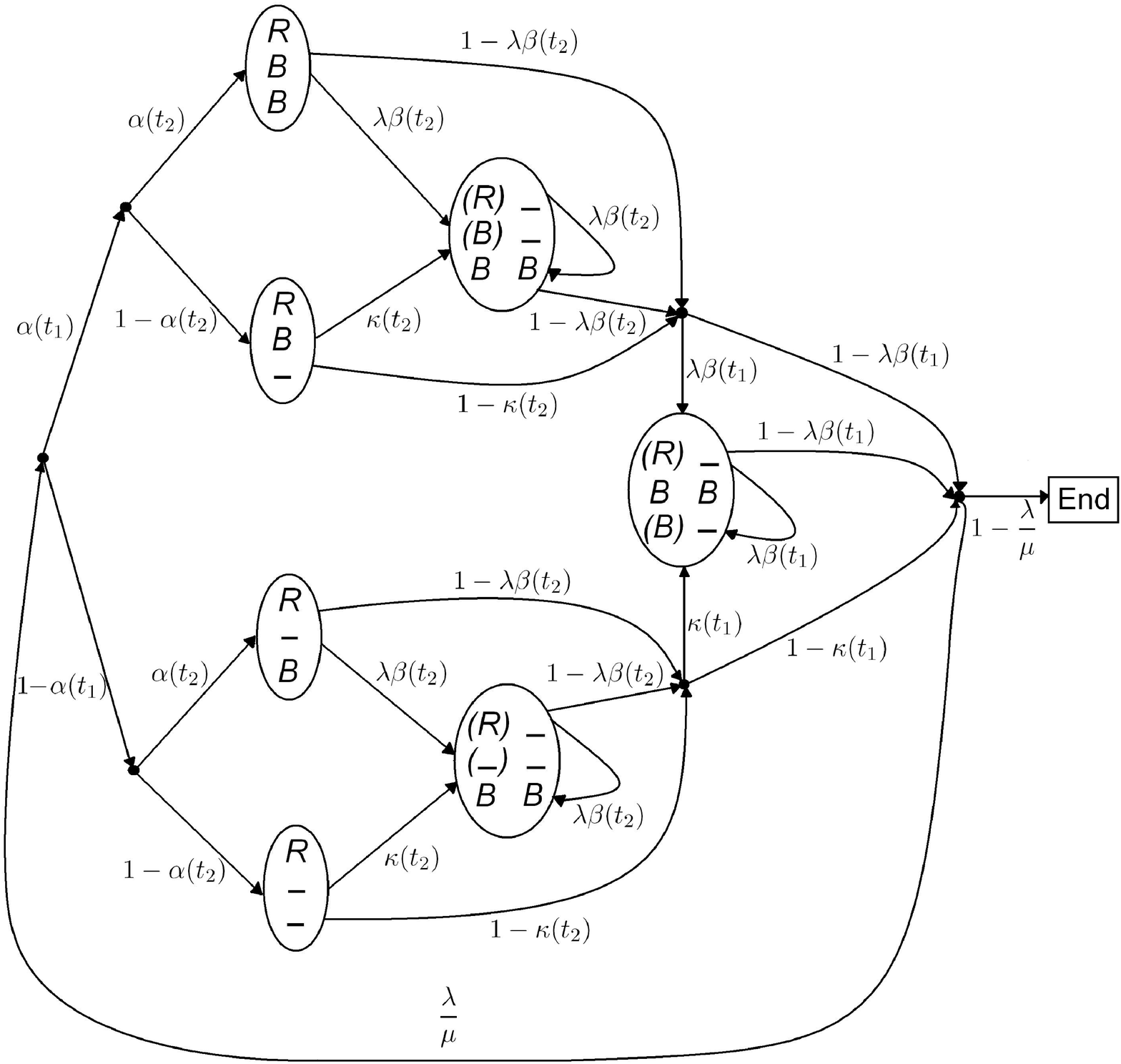}
 \caption{\emph{Alignment Markov chain for a star tree with two leaves and sequences evolving under the TKF91 evolution model as described by \citet{HolmesBruno}. The bubbles represent the states of the chain. \textsf{R} stands for a base on the ancestral sequence and \textsf{B} for a base on any of the observed sequences. Letters in brackets appear on insertion states and stand for the last event recorded on each sequence. The left column in insertion states is the true state whereas the right one is its representation in the alignment. There are two different states to represent an insertion on the second sequence since they have also to represent the fate (conservation or deletion) of the root nucleotide in the first sequence. That is because transitions from an insertion on the second sequence to an insertion on the first sequence depend on the fate of the last root nucleotide on the first sequence. Insertions on the second sequence are written in the alignment before insertions on the first sequence. The probabilities on the edges are those of (\ref{TKFprocess}) and $\alpha(t)=e^{-\mu t}$.}}
\label{mapa}
\end{figure}
So one could say that insertions to the root sequence break the Markov
dependence between alignment columns. Also, the order of the
insertions between two homologous positions is irrelevant, the
only important fact being which positions are homologous to which (see for instance the multiple alignment in Figure \ref{star} where the insertion columns are completely exchangeable). Then, the interesting objet is not the alignment but the homology
structure, essentially an alignment of homologous positions with
specification of the number of insertions on each sequence between
any two homologous positions.
The homology structure can be described in terms of the
nucleotides at the root sequence. Indeed the homology structure is
just the sequence of root positions in which we specify, for each
ancestral residue, its fate (whether it has survived or been deleted) and all
the insertions occurred to its right in each one of the observed sequences (see Figure \ref{hom_str} for an example). The homology structure is, as the \emph{bare} alignment, a reconstruction of the indel process of a set of sequences.\\
\begin{figure}
\begin{center}
{\scriptsize \textbf{\emph{Bare} alignment} \hspace{1.2cm} \textbf{Homology structure}\vspace{0.4cm}\\
\begin{tabular}{rl}
$X^1$:&      \texttt{BBB\textcolor[gray]{0.5}{BB-}B}\\
$X^2$:&      \texttt{B-B\textcolor[gray]{0.5}{---}B}\\
$X^3$:&      \texttt{BB-\textcolor[gray]{0.5}{---}B}\\
$X^4$:&      \texttt{BBB\textcolor[gray]{0.5}{--B}B}\\
$X^5$:&      \texttt{-BB\textcolor[gray]{0.5}{---}B}\\
$X^6$:&      \texttt{BBB\textcolor[gray]{0.5}{---}B}
\end{tabular}}
\hspace{0.8cm}
{\scriptsize\begin{tabular}{rl|l|l|l|}
$X^1$:&      1\,0&1\,0&1\,2&1\,0\\
$X^2$:&      1\,0&0\,0&1\,0&1\,0\\
$X^3$:&      1\,0&1\,0&0\,0&1\,0\\
$X^4$:&      1\,0&1\,0&1\,1&1\,0\\
$X^5$:&      0\,0&1\,0&1\,0&1\,0\\
$X^6$:&      1\,0&1\,0&1\,0&1\,0
\end{tabular}}
\end{center}
\caption{\emph{The bare alignment and the associated homology structure for the multiple alignment of Figure \ref{star}. There are two columns for each position in the ancestral sequence: in the first one, $1$ stands for a nucleotide that has been conserved and a $0$ for a nucleotide that has been deleted; the second column represents the number of insertions in each sequence to the right of the ancestral nucleotide. In the homology structure there is no artificial order between insertions in different sequences.}}\label{hom_str}
\end{figure}
In the TKF91 indel model, evolution on each \emph{link} is independent of evolution on other \emph{links} (see Thorne {\em et al.}, 1991). That is why the homology structure under these models can be described as a sequence of i.i.d. random variables as we will see in the next section.\\

\subsection{The homology structure on a star tree}
Consider a $k$-star phylogenetic tree $\T$ with branches lengths
$t_1,\dots,t_k$. The homology structure of the sequences related
by $\T$ is a sequence of independent and identically distributed
random variables $\{\e_n\}_{n\geq 1}$. The variable $\e_n$ represents the fate of the $n$-th ancestral sequence character (or fragment, if we consider fragment indel evolution models). Its distribution will depend on the chosen indel evolution model. Under the TKF91 indel evolution model $\{\e_n\}_{n\geq 1}$ is a sequence of i.i.d. random variables on
$$\E^k=\left\{ (e(1),e(2))=(\delta^{1:k},a^{1:k}) \,|\,\delta^i\in\{0,1\} ,\, a^i\geq 0,\, i=1,\dots,k \right\}.$$
The first column of $\e_n$ corresponds to the
homologous positions to the $n$-th ancestral character. If it is conserved in
sequence $i$, $i=1,\dots,k$, then $\e^i_n(1)=1$, else
$\e^i_n(1)=0$. It is possible for an ancestral character to have
been deleted in all the observed sequences ($\e_n(1)=0_k$, where $0_k$ stands for the $k$-dimensional vector with all
components equal to 0).
The second column of $\e_n$ represents the number of insertions on
the observed sequences between the $n$-th and the $(n+1)$-th
ancestral sequence characters. It is possible to have none
insertions in any of the observed sequences between two homologous
positions ($\e_n(2)=0_k$). See Figure~\ref{hom_str} for an example of an homology structure.

Due to the branch independence, the law of $\e_n$, $n\geq 1$, under the TKF91 indel model, is given by
\begin{equation}\label{loi_e}
\mathbb{P}_{\lambda}\left(\e_n\!=\!(\delta^{1:k},a^{1:k})\right)\!=\!\prod_{i=1}^k
\!\left(q^H_{a^i+1}(t_i)\right)^{\1{\delta^i=1}}\!
\left(q^N_{a^i}(t_i)\right)^{\1{\delta^i=0}}\!\!,\,\,\,\,\, (\delta^{1:k}\!,a^{1:k})\!\in \!\E^k\!.
\end{equation}
Conditionally to the result of the indel process (the \emph{bare} alignment), nucleotides on the observed sequences are emitted according to some substitution process. In practice, most nucleotide substitution processes are described by a continuous time Markov chain defined on $\A$ and depending on the branches lengths (see Felsenstein, 2004, for instance). Let us note $\nu$ the stationary law of this process and $p_t(\cdot,\cdot)$ the transition probability matrix for a transition time $t>0$. Then, for $n\geq1$, if $\e_n=(\delta^{1:k},a^{1:k})$, $r=\sum_{i=1}^k \delta^i$ nucleotides are emitted in the conserved positions according to the joint probability distribution $h_J$, $J=\{i|\delta^i=1\}$, on ${\cal A}^r$, with
\begin{equation}\label{subs_Markov}
h_{\{i_1,\dots,i_r\}}(x^{i_1},\dots,x^{i_r})=\sum_{R\in \A} \nu(R)\prod_{j=1}^{r}p_{t_{i_j}}
(R,x^{i_j}),
\end{equation}
where $R$ represents the unknown ancestral nucleotide. Note that $h_J$ does not only depend on the cardinal of $J$, but also on its elements via the branches lengths $\{t_i\}_{i=1,\dots,k}$. In the inserted positions, $\sum_{i=1}^k a^i$ nucleotides are emitted independently and identically distributed according to the probability distribution $f(\cdot) =\nu(\cdot).$

In classical substitution processes there is independence between the different sites of the ancestral sequence. That means that conditionally on $\{\e_n\}_{n\geq 1}$, the emissions of nucleotides on the observed sequences at different instants (positions of the ancestral sequence) are independent and equally distributed as described below.

\section{The multiple-hidden i.i.d. model on a star tree}

We present in this section the \emph{multiple-hidden i.i.d.} model, where \emph{multiple} refers to the number ($>2$) of observed sequences and \emph{i.i.d.} to the nature of the hidden process, by analogy to the name of the pair-hidden Markov model. The homology structure of $k$ sequences evolving under the TKF91 indel evolution model and a particular substitution model, as described in the precedent section, is a particular parametrization of this model.

Consider a sequence of i.i.d. random variables $\{\e_n\}_{n\geq 1}$ on the state space
$$\E^k=\left\{ (e(1),e(2))=(\delta^{1:k},a^{1:k}) \,|\,\delta^i\in\{0,1\} ,\, a^i\in\N\,\, i=1,\dots,k \right\}$$
with distribution $\pi$. 

The process $\{\e_n\}_{n\geq 1}$ generates a random walk
$\{Z_n\}_{n \geq 0}$ with values on $\N^k$ by letting $Z_0
=0_k$ and $Z_n =\sum_{1\leq j\leq n}[\e_j(1)+\e_j(2)]$ for
$n\geq 1$. The coordinate random variables corresponding to $Z_n$
at position $n$ are denoted by $(Z^1_n,\dots,Z^k_n)$ ({\it i.e.}
$Z_n=(Z^1_n,\dots,Z^k_n)$). In the homology structure context they represent the length of each observed sequence up to position $n$ on the ancestral sequence.

Let us now describe the emission of the observed sequences which take values on a finite alphabet ${\cal A}$. We distinguish to kinds of emissions, joint emissions across $k$ or a smaller number of sequences (corresponding to $\e_n(1)$) and single emissions (corresponding to $\e_n(2)$).
For $n\geq1$, if $\e_n=(\delta^{1:k},a^{1:k})$ then a vector of
$r=\sum_{i=1}^k \delta^i$ r.v. is emitted according to some probability
distribution $h_J$, $J=\{i|\delta^i=1\}$, on ${\cal A}^r$ and
$\sum_{i=1}^k a^i$ r.v. $\{ X^i_{1:a^i}, a^i \geq 1 \}$,
$i=1,\dots,k$, are emitted according to the following
scheme: $\{X^i_j\}^{i=1,k}_{j=1,a^i}$ are independent and
identically distributed from some probability distribution $f$ on
${\cal A}$.

Conditionally to the process $\{\e_n\}_{n\geq 1}$, the
random variables emitted at different instants are independent. The whole multiple-hidden i.i.d. model is described by the parameter $\theta=(\pi,\,
\{h_J\}_{J\subseteq K},\,f)\in \Theta$, where $K=\{1,\dots,k\}$.
We do not consider the branches lengths as a
component of the parameter and assume they are known.

The conditional distribution of the observations given an homology structure $e_{1:n}=(e_j)_
{1\leq j\leq n}=((\delta^{1:k}_j,a^{1:k}_j))_{1\leq j\leq n}$, writes
\begin{eqnarray} \label{conditional2}
&&\pr(\X_{1_k:Z_n}| \e_{1:n}=e_{1:n}, \{\e_m\}_{m>n},
\{X^i_{n_i}\}_{i\in K, n_i>Z^{i}_n})=
\pr(\X_{1_k:Z_n} |  \e_{1:n}=e_{1:n} ) \nonumber\\
&=& \prod_{j=1}^n \pr(\X_{Z_{j-1}+1_k:Z_j} |  \e_j=e_j ) \nonumber \vspace{-0.2cm}\\
&=&\prod_{j=1}^n \Big\{ h_{\{i|\delta_j^i=1\}}
\big(\{X^i_{Z^i_{j-1}+1}\}_{i|\delta^i_j=1}\big) \prod_{i=1}^k
\prod_{s=1}^{a^i_j}
 f\big(X^i_{Z^i_{j-1} +\delta^i_j +s}\big) \Big\}
\end{eqnarray}
where $1_k$ stands for the $k$-dimensional vector with all
components equal to 1 and $\X_{1_k:Z_n}= (X^1_{1:Z^{1}_n},
\dots, X^k_{1:Z^{k}_n})$. This notation can be confusing since it
is possible to have $Z^{i}_{j-1}+1_k > Z^{i}_{j}$ for some $i\in K$ and for some
$j\geq 1$. However when writing $\X_{Z_{j-1}+1_k: Z_{j}}$ we will only be
considering the variables corresponding to those sequences $i\in K$ for which
$Z^{i}_{j-1}+1_k \leq Z^{i}_{j}$.

The complete distribution $\pr$ is given by
\begin{eqnarray*}
&&\pr(\e_{1:n}=e_{1:n},\X_{1_k:Z_n})=  \pr(\X_{1_k:Z_n} | \e_{1:n}=e_{1:n} )  \pr(\e_{1:n}=e_{1:n} ) \\
&=& \pr(\X_{1_k:Z_n} | \e_{1:n}=e_{1:n} ) \prod_{j=1}^n
\pr(\e_j=e_j)= \pr(\X_{1_k:Z_n} | \e_{1:n}=e_{1:n} ) \prod_{j=1}^n
\pi(e_j)
\end{eqnarray*}
We denote by $\pr$ (and $\esp$) the induced probability distribution (and
corresponding expectation) on $\E^{\N}  \times (\A^{\N})^k $ and
$\theta_0=(\pi_0,\,
\{h_{0_J}\}_{J\subseteq K},\,f_0)$ the true parameter corresponding  to the distribution
of the observations (we shall abbreviate to $\pro$ and $\espo$ the
probability distribution and expectation under parameter
$\theta_0$).
\subsection{Observations and likelihoods}

As in the pair-HMM (see Arribas-Gil {\em et al.}, 2006) there are different
interpretations of what the observations represent on this model,
and thus different definitions for the log-likelihood of the
observed sequences $(X^1_{1:n_1},\dots,X^k_{1:n_k})$. However, the
difference with the pair-HMM is that in the multiple-hidden i.i.d. model we suppose that
the observed sequences are cut out of very much longer sequences
between known homologous positions. This implies that any
interpretation of what observations represent must assume that the
underlying process $\{\e_n\}_{n \geq 1}$ passes through the points $0_k$ and
$(n_1,\dots,n_k)$.

One may consider that what we observe are sequences that have
evolved from an ancestral sequence of length $n$ so that the
likelihood should be $\pr(\X_{1_k:Z_n})$ $=\pr(\X_{1_k:Z_n},Z_n)$. This
term is computed by summing, over all possible homology structures
from an ancestral sequence of length $n$, the probability of
observing the sequences and a homology
structure. 

Let us define $\E_{n_1,\dots,n_k}$ the set of all possible
homology structures of $k$ sequences of lengths $n_1,\dots,n_k$:
\begin{equation}
\E_{n_1,\dots,n_k}=\{ e\in(\E^k)^n ;\,\, n\in \N, \,\,
\sum_{j=1}^{n}|e_j|=(n_1,\dots,n_k) \}.
\end{equation}
For any homology structure $e\in \E_{n_1,\dots,n_k}$, if $e \in
(\mathcal{E}^k)^n$, then $n$ is the length of the path $e$ and is
denoted by $|e|$. In the homology structure context, $|e|$ stands for the length of the ancestral sequence.
So we have
$$\pr(\X_{1_k:Z_n})=\pr(\X_{1_k:Z_n},Z_n)=\sum_{e  \in \E_{Z_n};|e|=n} \pr
(\e_{1:n}=e,\X_{1_k:Z_n}).$$ Then, we would define the
log-likelihood $\ell_n(\theta)$ as
\begin{equation}
\label{lt} \ell_n  (\theta)  =  \log  \pr  (\X_{1_k:Z_n}) ,  \quad n
\geq  1 .
\end{equation}
But since  the underlying process $\{Z_n\}_{n \geq  0}$ is not
observed, the   quantity  $\ell_n(\theta)$  is   not  a measurable
function   of  the observations. More precisely, the length $n$ at
which the observation is made is not observed itself. Though, if
one decides that $(X^1_{1:n_1},\dots,X^k_{1:n_k})$ corresponds to
the observation of  the  emitted sequences  at  a  point  of the
hidden  process $Z_n=(Z^i_n)_{i=1,\dots,k}$ and some {\it unknown}
``ancestral length" $n$, one does not use $\ell_n(\theta)$ as a
log-likelihood, but rather
\begin{equation}
\label{qt} w_n (\theta) =  \log \Q (\X_{1_k:Z_n}) , \quad n \geq 1
\end{equation}
where for any integers $n_i, i=1,\dots,k$
\begin{equation}
\Q(X^1_{1:n_1},\dots,X^k_{1:n_k}) =\pr ( \exists m \geq 1, Z_m
=(n_1,\dots,n_k) ; X^1_{1:n_1},\dots,X^k_{1:n_k} ).
\end{equation}
In other words, $\Q$ is the probability of the observed sequences
under the assumption that the underlying process $\{\e_n\}_{n \geq
1}$ passes through the  point $(n_1,\dots,n_k)$. But the length of
the ancestral sequence remains unknown when computing $\Q$. This
gives the formula:
\begin{equation}  \label{Qnm2}
\Q(X^1_{1:n_1},\dots,X^k_{1:n_k})  =\sum_{e  \in
\E_{n_1,\dots,n_k}} \pr
(\e_{1:|e|}=e,X^1_{1:n_1},\dots,X^k_{1:n_k}).
\end{equation}
Let us stress that we have
$$
w_n(\theta) =\log \pr(\exists m\geq  1, Z_m
=(Z^{i}_n)_{i=1,\dots,k}; X^1_{1:Z^{1}_n}, \dots,
X^k_{1:Z^{k}_n}), \quad n\geq 1,
$$
meaning that  the length of the ancestral sequence is not
necessarily $n$, but is in
fact unknown.

In the homology structure context, $\Q$  is  the  quantity  that  is  computed  by the multiple
alignment algorithms (see for instance Holmes and Bruno, 2001, Steel and Hein, 2001, or
Lunter {\em et al.}, 2003) and which is used as likelihood in
biological applications. The more extended application is to use
this quantity to co-estimate alignments and phylogenetic trees
in a Bayesian framework via MCMC calculations (cf. Fleissner {\em et al.}, 2005; Lunter {\em et al.}, 2005; Novák {\em et al.}, 2008). Indeed, algorithms that perform this joint estimation compute, at each iteration, the likelihood of sequences for a given phylogenetic tree. Thus,
asymptotic properties of the criterion $\Q$ and consequences on
asymptotic properties of the estimators derived from $\Q$
are of primarily interest.

We will look for asymptotic results for $n\to \infty$. We need to
establish some kind of relationship between $n$ and
$n_1,\dots,n_k$, to derive asymptotic results for $n_i \to
\infty$. From our definition of the multiple-hidden i.i.d. model, it is clear
that it does not exist a deterministic relationship between the
length of the hidden sequence and the lengths of the observed
sequences. However, in the multiple alignment problem, a natural assumption is that very big
insertions and deletions occur rarely and thus the length of the
root sequence should be equivalent to the lengths of the observed sequences. In fact we have the following result.\\
\begin{lemma}\label{esp=1}
In the multiple-hidden i.i.d. model on a star tree under the TKF91 indel evolution process, that is, when $\pi$ is the distribution given by (\ref{loi_e}), for any $\lambda >0$ we have $Z_n^i \sim n$, $i=1,\dots,k$, $\mathbb{P}_{\lambda}$-almost surely.
\end{lemma}
\begin{proof}.
For all $i=1,\dots,k$ and for all $n\geq 1$ we have that
$$Z_n^i= \sum_{j=1}^n (\e^i_j(1) + \e^i_j(2))$$
where $\{\e^i_j\}_{j\geq 1}$ are i.i.d. Moreover, from (\ref{newprocess}) we have, for any $\lambda >0$
\begin{multline*}
\mathbb{E}_{\lambda}\,[\e^i_j(1) + \e^i_j(2)]\\
=\displaystyle\sum_{m\geq 1} m \!\left\{\mathbb{P}_{\lambda}(\e^i_j(1) + \e^i_j(2)=m,
\e^i_j(1)=0)\!+\!\mathbb{P}_{\lambda}(\e^i_j(1) + \e^i_j(2)=m,
\e^i_j(1)=1)\right\}\\
=\displaystyle\sum_{m\geq 1} m \left\{ q_m^N(t_i) + q_m^H(t_i)\right\}\hfill \phantom{w}\\
=\displaystyle\sum_{m\geq 1} m \left\{\!\!\left(\!\frac{1}{1+\lambda
t_i}-e^{-\lambda t_i}\!\right)\frac{1}{1+\lambda
t_i}\left(\!\frac{\lambda t_i}{1+\lambda t_i}\!\right)^{m-1} \!\!\!+\!
e^{-\lambda t_i}\frac{1}{1+\lambda t_i}\left(\!\frac{\lambda
t_i}{1+\lambda t_i}\!\right)^{m-1}\!\right\}\\=1.\hfill \phantom{w}
\end{multline*}
Now the result holds from the strong law of large numbers.
\end{proof}\\

According to this lemma, under the TKF91 indel evolution model, asymptotic results for $n \to \infty$
will imply equivalent ones for $n_i \to \infty,\, i=1,\dots,k$. Let us establish an assumption to get the same result for the general multiple-hidden i.i.d. model.
\begin{assumption} In the multiple-hidden i.i.d. model on a star tree $\esp\,[\e_n(1)+\e_n(2)]=1_k$, for $n\geq 1$, for any $\theta\in\Theta$.
\end{assumption}

\subsection{The case of two sequences}
Let us consider the case in which $k=2$. It is clear that the general multiple-hidden i.i.d. model and the pair-HMM are different in this case. However, in the context of the alignment of two sequences evolving under the TKF91 model, the two models are equivalent. In fact, in the pairwise alignment we consider that one of the sequences is the ancestor of the other one, but since the TKF91 model is time reversible, this is equivalent to consider that both sequences evolve from a common unknown ancestor.

First of all, let us remark that the likelihood ($\Q$) of two sequences
$x_{1:n}$ and $y_{1:m}$ is the same under the two models. Let $t$ be the evolution time between both sequences, that is, the sum of the evolution times between the root and each one of the sequences, $t_1+t_2$, in the multiple alignment setup. Consider for the pair-HMM the following transition matrix:\vspace{-0.1cm}
\begin{multline}\label{mult_2seq}
\begin{array}{ccccc}
\qquad \qquad \quad  D & &\qquad\qquad \quad \quad \quad H & &\qquad \quad \quad\quad V
\end{array}\\
\begin{array}{c}
D\vspace{0.425cm}\\
H\vspace{0.425cm}\\
V \end{array}
\left( \begin{array}{ccccc}
\displaystyle{\frac{\alpha(t)}{(1+\lambda t)}} & &
\displaystyle{\frac{1-\alpha(t)}{(1+\lambda t)}}& &
\displaystyle{\frac{\lambda t}{(1+\lambda t)}}\vspace{0.175cm}\\
(1-\kappa(t))\alpha(t)& &(1-\kappa(t))(1-\alpha(t))
& & \kappa(t)\vspace{0.175cm}\\
\displaystyle{\frac{\alpha(t)}{(1+\lambda t)}} & &
\displaystyle{\frac{1-\alpha(t)}{(1+\lambda t)}}& &
\displaystyle{\frac{\lambda t}{(1+\lambda t)}}
\end{array}\right)\qquad\quad \phantom{Q}
\end{multline}\\
where $D$, $H$ and $V$ stand for diagonal, horizontal and vertical movements respectively, with the notations of \citet{Argamat}, and
$\alpha(t)=e^{-\lambda t}$, $\kappa(t)=1-\frac{\lambda t}{(1+\lambda t)(1-\alpha(t))}$. It is easy to show that the probability of an homology structure (under the multiple-hidden i.i.d. model) is just the sum of the probabilities of all possible alignments (under the pair-HMM) leading to that homology structure. Then, the sum over all possible alignments and all possible homology structures of two sequences is equivalent.

Finally, note that for the transition matrix in (\ref{mult_2seq})
the stationary probabilities of insertions
and deletions are the same, that is $p=q$ with the notations of \citet{Argamat}. That means that we are in the
case where the \emph{main direction} of the alignment, that is, its expectation under the pair-HMM, is always the straight line from $(0,0)$ to $(n,n)$ for every value of the parameter. This is also the case in the multiple-hidden i.i.d. model as we have shown in Lemma \ref{esp=1}.

\section{Information divergence rates in the star tree model}\label{Inf_Div}
\subsection{Definition of Information divergence rates}
In this section we prove the convergence of the normalized \emph{log-likelihoods}
$\ell_n(\theta)$ and $\omega_n(\theta)$.
Let us note
\begin{multline*}
\Theta_{0} =\left\{\theta \in \Theta \,\, | \,\, \pi(e)>0, \, \,
h_J(x^{1:|J|})>0, \,\, f(y)>0, \right.\\
\left.\,\forall e \in \E^{k}, \,\,\forall x^{1:|J|} \in{\cal A}^{|J|},\,\, \forall J\subseteq K,\,\,\forall y\in{\cal
A}\right\}.
\end{multline*}
We shall always assume that $\theta_{0}\in \Theta_{0}$.

\begin{thm}
\label{thdivergence2} The following holds for any
$\theta\in\Theta_0$:
\begin{itemize}
\item[i)] $ n^{-1} \ell_{n} (\theta)$ converges $\pro$-almost
surely and in $\mathbb{L}_1$, as $n$ tends to infinity to
$$
\ell (\theta ) = \lim_{n\rightarrow \infty}\frac{1}{n}\espo
\left(\log \pr(\X_{1_k:Z_n}) \right) = \sup_{n}\frac{1}{n}\espo
\left(\log \pr(\X_{1_k:Z_n}) \right).
$$
\item[ii)] $ n^{-1} w_{n} (\theta)$ converges $\pro$-almost surely
and in $\mathbb{L}_1$, as $n$ tends to infinity to
$$
w (\theta ) = \lim_{n\rightarrow \infty}\frac{1}{n}\espo
\left(\log \Q(\X_{1_k:Z_n}) \right) = \sup_{n}\frac{1}{n}\espo
\left(\log \Q(\X_{1_k:Z_n}) \right).
$$
\end{itemize}
\end{thm}
Using the terminology of \citet{Argamat} we then define Information divergence rates:
\begin{defi} $\forall \theta \in \Theta_0, \;
D (\theta \vert \theta_0)=w (\theta_0)-w (\theta) \quad \text{and} \quad
D^{*}(\theta \vert \theta_0)=\ell (\theta_0)-\ell (\theta).$
\end{defi}
We recall that $D^{*}$ is what is usually called the Information divergence rate in Information Theory:
it is the limit of the normalized Kullback-Leibler divergence between the distributions of
the observations at the true parameter value and another parameter value. However, we also call
$D$ an  Information divergence rate  since $\Q$ may   be interpreted  as a
likelihood.\\
\\
\begin{proof} { \bf of Theorem \ref{thdivergence2}.}
This proof is similar to the proof of Theorem 1 in \citet{Argamat}. We
shall use the following version of the sub-additive ergodic
Theorem due to \citet{Kingman} to prove point {\it i)}.
A similar proof may be written for {\it ii)} and is left to the reader.\\
Let $(W_{s,t})_{0\leq s <t}$
be a sequence of random variables such that
\begin{enumerate}
\item For all $m<n$, $W_{0,n}\geq W_{0,m}+ W_{m,n}$, \item For all
$l>0$, the joint distributions of $(W_{m+l,n+l})_{0\leq m <n}$ are
the same as those of $(W_{m,n})_{0\leq m <n}$, \item
$\espo(W_{0,1})
> -\infty$.
\end{enumerate}

Then $\lim_{n}\! n^{\mbox{\tiny$-1$}}  W_{0,n}$ exists almost
surely. If moreover the sequences $(\!W_{m+l,n+l})_{l>0}$ are
ergodic, then the limit is almost surely deterministic and equals
$\sup_{n}\! n^{\mbox{\tiny$-1$}} \mathbb{E}_0(\!W_{0,n}\!)$. If moreover
$\espo(W_{0,n})\leq An$,  for some constant $A\geq 0$ and all $n$,
then the convergence holds in $\mathbb{L}_1$.

We apply this theorem to the process
$$
W_{m,n}= \log \pr(\X_{Z_m+1_k:Z_n}), \quad 0\leq m<n.
$$
Note that since $Z_0=0_k$ is deterministic, we have
$W_{0,n} = \log \pr(\X_{1_k:Z_n})$. Super-additivity (namely point
1.) follows since for any $0\leq m<n$,
\begin{multline*}
\pr(\X_{1_k:Z_n})= \sum_{\substack{e  \in \E_{Z_n}\\|e|=n}}
\pr (\e_{1:n}=e_{1:n},\,\X_{1_k:Z_n})\\
\phantom{w}\geq \sum_{\substack{e  \in \E_{Z_m}\\|e|=m}} \sum_{\substack{e'
\in \E_{Z_n-Z_m}\\|e'|=n-m}}
\pr (\e_{1:m}=e_{1:m},\, \e_{m+1:n}=e'_{1:n-m}, \,\X_{1_k:Z_n})\hfill \phantom{w}\\
\geq \sum_{\substack{e  \in \E_{Z_m}\\|e|=m}} \sum_{\substack{e'
\in \E_{Z_n-Z_m}\\|e'|=n-m}} \pr (\e_{m+1:n}=e'_{1:n-m},\,
\X_{Z_m+1_k:Z_n})
\times \pr (\e_{1:m}=e_{1:m},\, \X_{1_k:Z_m})\\
=\pr(\X_{1_k:Z_m}) \times \pr(\X_{Z_m+1_k:Z_n})
\end{multline*}
so that we get $ W_{0,n} \geq W_{0,m} +W_{m,n}$, for any $0\leq m<n$.

To understand the distribution of $(W_{m,n})_{0\leq m <n}$, note
that $W_{m,n}$ only depends on trajectories of the random walk
going from the point $(Z^{1}_{m},\dots,Z^{k}_m)$  to the point
$(Z^{1}_{n},\dots,Z^{k}_n)$ with length $n-m$. Since the variables
$(\e_n)_{n\geq 1}$ are i.i.d., one gets that the distribution of
$(W_{m,n})$ is the same as that of
$(W_{m+l,n+l})$ for any $l$, so that point $2.$ holds.

Point $3.$ comes from:
\begin{multline*}
\pr(\X_{1_k:Z_1})=\sum_{\substack{e  \in \E_{Z_1}\\|e|=1}} \pr
(\e_{1}=e)\,\pr(\X_{1_k:Z_1}|\e_{1}=e)\\
=\sum_{\substack{e  \in \E_{Z_1}\\|e|=1}}\pi(e)
\left\{h_{\{i|\delta_1^i=1\}} \big(\{X^i_1\}_{i|\delta^i_1=1}\big)
\prod_{i=1}^k \prod_{s=1}^{a^i_1}
 f\big(X^i_{\delta^i_1 +s}\big) \right\}>0
\end{multline*}
$\pro$-almost surely, since $\theta \in \Theta_0$, provided that $Z^i_1\geq 1$ for some $i\in K$.
So $\espo(W_{0,1} ) = \espo \log \pr(\X_{1_k:Z_1})>-\infty$.

Let us fix $0\leq m <n$. The proof that
$W^{s,t}=(W_{m+l,n+l})_{l>0}$ is ergodic is the same as that of
\citet{Leroux} (Lemma 1). Let $T$ be the shift operator, so that if
$u=(u_l)_{l\geq 0}$, the sequence $Tu$ is defined by
$(Tu)_{l}=(u)_{l+1}$ for any $l\geq 0$. Let $B$ be an event which
is $T$-invariant. We need to prove that $\pro(W^{m,n} \in B)$
equals $0$ or $1$. For any integer $i$, there exists a cylinder
set $B_i$, depending only on the coordinates $u_l$ with $-j_i \leq
l \leq j_i$ for some sub-sequence $j_i$, such that $\pro(W^{m,n}
\in B\Delta B_{j_i})\leq 1/2^i$. Here, $\Delta$ denotes the
symmetric difference between sets. Since $W^{m,n}$ is stationary
and $B$ is $T$-invariant:
\begin{eqnarray*}
\pro\left(W^{m,n}  \in B\Delta  B_{j_i}\right)=
\pro\left(T^{2j_i}W^{m,n} \in
  B\Delta B_{j_i}\right)
=\pro\left(W^{m,n} \in B\Delta T^{-2j_i}B_{j_i}\right).
\end{eqnarray*}
Let $\tilde{B}=\cap_{i\geq 1}\cup_{h\geq i}T^{-2j_h}B_{j_h} $.
Borel-Cantelli's Lemma leads to $\pro(W^{m,n} \in B\Delta
\tilde{B})=0$, so that $\pro(W^{m,n}  \in  B)=\pro(W^{m,n}  \in
\tilde{B})=\pro(W^{m,n}  \in  B  \cap  \tilde{B})$.  Now,
conditional on $(\e_n)_{n\in\N}$, the random variables
$(W_{m+l,n+l})_{l> 0}$ are strongly mixing. Indeed
$W_{m+l,n+l}$ only depends on a finite number of other $(W_{m+k,n+k})$, $k>0$, namely
$(W_{m+k,n+k})_{k=max(1,m+l-n+1),\dots,n+l-m}$. Then the $0-1$
law for strongly mixing processes (see Sucheston, 1963) implies that for any fixed sequence $e$ with
values in $(\E^k)^{\N}$, the probability $\pro(W^{m,n} \in
\tilde{B}\vert (\e_n)_n=e)$ equals $0$ or $1$, so that
$$
\pro\left(W^{m,n} \in \tilde{B}\right)=\pro\left((\e_n)_n \in
C\right)
$$
where $C$ is the set of sequences $e$ such that $\pro(W^{m,n} \in
\tilde{B}\vert (\e_n)_n=e)=1$. But it is easy to see that $C$ is
$T$-invariant. Indeed, if $e\in C$ then, since $W^{m,n}$ is
stationary and $\tilde{B}$ invariant,
\begin{multline*}
1=\pro(W^{m,n} \in
\tilde{B}\vert (\e_n)_n=e)=\pro(TW^{m,n} \in \tilde{B}\vert
(\e_n)_n=Te)\\=\pro(W^{m,n} \in \tilde{B}\vert (\e_n)_n=Te)
\end{multline*}
so that $Te\in C$. Now, since $(\e_n)_{n\geq 1}$ is an i.i.d. process, it is ergodic so
$\pro\left((\e_n)_n \in C\right)$ equals $0$ or
$1$. This concludes the proof of ergodicity of the sequence $W^{m,n}$.

To end with, note that for any $n\geq 0$, the random variable
$W_{0,n}$ is non positive,  ensuring the convergence of
$\{n^{-1}W_{0,n}\}$ in $\mathbb{L}_1$.
\end{proof}

\subsection{Divergence properties of Information divergence rates}
Information divergence rates should be non negative: this is
proved below. They also should be  positive for parameters  that
are different  than the true  one: we only prove it in a
particular subset of the parameter set. Let us define the set
$$
\Theta_{marg}=\left\{\theta\in\Theta_0\;:\;h^i_{J}=f, \,\forall
J\subseteq K,\,\forall i\in J\right\}.
$$
where $h^i_{J}$ denotes the $i$-th marginal of $h_J$.
\begin{thm}
\label{contrast2} Information divergence rates satisfy:
\begin{itemize}
\item For all $\theta \in \Theta_0$, $D(\theta \vert \theta_0)
\geq 0$ and $D^{*} (\theta \vert \theta_0) \geq 0$. \item If
$\theta_0$ and $\theta$ are in $\Theta_{marg}$, $D(\theta \vert
\theta_0) > 0$ and $D^{*}(\theta \vert \theta_0) > 0$ as soon as
$f\neq f_0$.
\end{itemize}
\end{thm}

Note that from Assumption 1 the expectation of $\e_n(1)+\e_n(2)$, $n\geq 1$, is the same for any value of the parameter. Thus, we can not establish the positivity of the information divergence rates for values of $\theta$ for which the expectation of the hidden process is different than under $\theta_0$, as it is done for pair-HMMs (Theorem 2 of Arribas-Gil {\em et al.}, 2006).

Also note that when we consider classical markovian substitution processes for the emission laws, as described in (\ref{subs_Markov}), the parameter always lies in $\Theta_{marg}$, since the marginal emission distributions are equal to the stationary distribution of the Markov process.\\
\\
\begin{proof}. Since for all $n$,
$$
\espo \left(\log \pro(\X_{1_k:Z_{n}}) \right)- \espo \left(\log
\pr(\X_{1_k:Z_{n}}) \right)
$$
is a Kullback-Leibler divergence, it is non negative, and the
limit
$D^{*}(\theta \vert \theta_0)$ is also non negative.

Let us prove that $D(\theta \vert \theta_0)$ is also  non
negative. To compute  the value  of the expectation  $\espo [w_{n}
(\theta )]$, note that the set of all possible values of $Z_n$ is
$\N^k$. Then,
\begin{multline*}
\espo [w_{n} (\theta)]\\=\!\sum_{(n_1,\dots,n_k)\in \N^k }
\sum_{(x^i_{1:n_i})_{i=1,\dots,k}} \!\!\pro\big(Z_{n} =(n_1,\dots,n_k),
X^1_{1:n_1}=x^1_{1:n_1},\dots,X^k_{1:n_k}=x^k_{1:n_k} \big) \\
\phantom{=\sum_{(n_1,\dots,n_k)\in \N^k }
\sum_{(x^i_{1:n_i})_{i=1,\dots,k}} \pro\big(Z_{n} =(n_1,\dots,n_k),
X^1_{1:n_1},,}\times\log \Q(x^1_{1:n_1},\dots,x^k_{1:n_k} ).
\end{multline*}
Now, by definition,
$$
D\left(\theta \vert \theta_0 \right ) = \lim_{n\rightarrow
+\infty} \frac{1}{n} \espo \left(\log
\frac{\Qo(\X_{1_k:Z_n})}{\Q(\X_{1_k:Z_n})}\right).
$$
By using Jensen's inequality,
\begin{equation*}
\espo \!\left(\!\log \frac{\Q(\X_{1_k:Z_n})}{\Qo(\X_{1_k:Z_n})}\!\right)\!
\leq \log \espo\!
\left(\frac{\Q(\X_{1_k:Z_n})}{\Qo(\X_{1_k:Z_n})}\right)\!= \log
\espo\! \left[
\espo\!\left(\frac{\Q(\X_{1_k:Z_n})}{\Qo(\X_{1_k:Z_n})}\right)\!\big|
Z_n \right]\!.
\end{equation*}
Now, for all $(n_1,\dots,n_k )\in \N^k$
\begin{multline*}
 \espo\left(\frac{\Q(\X_{1_k:Z_n})}{\Qo(\X_{1_k:Z_n})}\big|
 Z_n=(n_1,\dots,n_k)\right)\\
=\!\sum_{(x^i_{1:n_i})_{i=1}^k}\!\!\!\pro\big(Z_{n}\!=\!(n_1,\mbox{\tiny $\dots$},n_k),
X^1_{1:n_1}\!\!=x^1_{1:n_1},\mbox{\tiny $\dots$},X^k_{1:n_k}\!\!=x^k_{1:n_k}
\big) \frac{\Q(x^1_{1:n_1},\mbox{\tiny $\dots$},x^k_{1:n_k}
)}{\Qo(x^1_{1:n_1},\mbox{\tiny $\dots$},x^k_{1:n_k} )}\\
 \stackrel{(a)}{\leq}\sum_{(x^i_{1:n_i})_{i=1}^k} \pr\big(\exists m\geq  1, Z_m
 =(n_1,\dots,n_k),\,
 X^1_{1:n_1}=x^1_{1:n_1},\dots,X^k_{1:n_k}=x^k_{1:n_k}\big)\phantom{\stackrel{(a)}{\leq}\sum_{(x^i_{1:n_i})_{i=1,\dots,k}} }\\
 \phantom{\stackrel{(a)}{\leq}\sum_{(x^i_{1:n_i})_{i=1,\dots,k}} \pr\big(\exists m\geq  1, Z_m
 =xxxxx)}=\pr\big(\exists m\geq  1, Z_m =(n_1,\dots,n_k)\big)\leq 1
\end{multline*} where $(a)$ comes from expression (\ref{Qnm2}).
Thus, $\espo \left[
\espo\left(\frac{\Q(\X_{1_k:Z_n})}{\Qo(\X_{1_k:Z_n})}\right)\big|
Z_n \right] \leq 1$, and
\begin{equation*}
\lim_{n\rightarrow
+\infty}\frac{1}{n}\left(w_{n}(\theta)-w_{n}(\theta_0) \right)\leq
\liminf_{n\rightarrow +\infty}\frac{1}{n} \log \espo \left[
\espo\left(\frac{\Q(\X_{1_k:Z_n})}{\Qo(\X_{1_k:Z_n})}\big| Z_n
\right)\right]\leq 0.
\end{equation*}
So finally
$$
\forall \theta \in \Theta_0, \;D(\theta \vert \theta_0) \geq 0.
$$

Let us now consider the case where $\theta_0$ and $\theta$ are in
$\Theta_{marg}$. Let us remark that for any $\theta \in
\Theta_{marg}$ we have
\begin{multline}\label{re}
\pr\big(Z_n=(n_1,\dots,n_k),\,X^1_{1:n_1}=x^1_{1:n_1}\big)\\
=\sum_{(x^i_{1:n_i})_{i=2}^k}\pr\big(Z_n
 =(n_1,\dots,n_k),\,X^1_{1:n_1}=x^1_{1:n_1},\dots,X^k_{1:n_k}=x^k_{1:n_k}\big)
 \phantom{\pr\big(Z_n=(n_1,\dots,n_k),\,X^1_{1:n_1}=x^1_{1:n_1}\big)}\\
=\sum_{\substack{e \in \E_{n_1,\dots,n_k}\\|e|=n}}
\sum_{(x^i_{1:n_i})_{i=2}^k}\pr(\e_{1:n}
 =e,\,
 X^1_{1:n_1}=x^1_{1:n_1},\dots,X^k_{1:n_k}=x^k_{1:n_k}) \phantom{\pr\big(Z_n=(n_1,\dots,n_k),\,X^1_{1:n_1}=x^1_{1:n_1}\big)}\\
=\sum_{\substack{e \in \E_{n_1,\dots,n_k}\\|e|=n}}\sum_{(x^i_{1:n_i})_{i=2}^k}\pr(\e_{1:n}
 =e)\,\pr(
 X^1_{1:n_1}=x^1_{1:n_1},\dots,X^k_{1:n_k}=x^k_{1:n_k}|\e_{1:n}
 =e)\phantom{wwwwww}\\
=\pr\big(Z_n=(n_1,\dots,n_k)\big)f^{\otimes n_1} (x^1_{1:n_1})
\end{multline}
where the last equality comes from (\ref{conditional2}). In the same
way, for any $\theta \in \Theta_{marg}$ we have that $\pr(\exists m\leq
1, Z_m=(n_1,\dots,n_k), X^1_{1:n_1}=x^1_{1:n_1})=\pr(\exists m\leq
1, Z_m=(n_1,\dots,n_k))f^{\otimes n_1} (x^1_{1:n_1})$. This is also
true for any other sequence $X^i_{1:n_i}$, $i=1\dots,k$. Then, using Jensen's inequality and definition \eqref{Qnm2},
\begin{multline*}
  \espo  \left(  \log
  \frac{\Q(\X_{1:Z_n}) } {\Qo(\X_{1:Z_n})}\right)  \\
=  \sum_{(n_1\mbox{\tiny $\dots$},n_k)  \in \N^k}  \sum_{(x^i_{1:n_i})_{i=1,\mbox{\tiny $\dots$},k}}  \pro\big(Z_n
=(n_1\mbox{\tiny $\dots$},n_k), \,
 X^1_{1:n_1}=x^1_{1:n_1},\mbox{\tiny $\dots$},X^k_{1:n_k}=x^k_{1:n_k}\big) \phantom{wwwwwwwwwwwwwww}\\\times \log
\frac{\Q(x^1_{1:n_1},\mbox{\tiny $\dots$},x^k_{1:n_k}) } {\Qo(x^1_{1:n_1},\mbox{\tiny $\dots$},x^k_{1:n_k})}\quad
\!\!\leq \! \!\!\sum_{(n_1\mbox{\tiny $\dots$},n_k)  \in \N^k}  \sum_{x^1_{1:n_1}}
 \pro\big(Z_n=(n_1,\mbox{\tiny $\dots$},n_k),\,X^1_{1:n_1}=x^1_{1:n_1}\big) \\
\times\log \mbox{\smalleq$\left(\sum_{(x^i_{1:n_i})_{i=2}^k} \!\!\! \frac{\pro\big(Z_n
=(n_1,\mbox{\tiny $\dots$},n_k),\,
 X^1_{1:n_1}=x^1_{1:n_1},\mbox{\tiny $\dots$},X^k_{1:n_k}=x^k_{1:n_k}\big)\Q(x^1_{1:n_1},\mbox{\tiny $\dots$},x^k_{1:n_k} ) }
{\pro\big(Z_n
=(n_1\mbox{\tiny $\dots$},n_k),\,X^1_{1:n_1}=x^1_{1:n_1}\big)  \Qo(x^1_{1:n_1},\mbox{\tiny $\dots$},x^k_{1:n_k}) }\!
\right)$} \\
\leq \sum_{(n_1\mbox{\tiny $\dots$},n_k)  \in \N^k}  \sum_{x^1_{1:n_1}}
 \pro\big(Z_n=(n_1,\mbox{\tiny $\dots$},n_k)\big) f_0^{\otimes n_1} (x^1_{1:n})\phantom{sssssssssssssssssssssssssssssssssssssssssss}\\  \phantom{sssssssssssssssssssssssss}\times \log \left(\frac{ \pr\big(\exists m\geq 1, Z_m=(n_1,\mbox{\tiny $\dots$},n_k)\big) f^{\otimes n_1} (x^1_{1:n_1}) } {\pro\big(Z_n
=(n_1,\mbox{\tiny $\dots$},n_k)\big) f_0^{\otimes n_1} (x^1_{1:n_1}) } \right)\!,
\end{multline*}
where the last inequality comes from (\ref{re}) and the
fact that $$\pro\big(Z_n =(n_1,\dots,n_k),\,
 X^1_{1:n_1}=x^1_{1:n_1},\dots,X^k_{1:n_k}=x^k_{1:n_k}\big) \leq \Qo(x^1_{1:n_1},\dots,x^k_{1:n_k}).$$
Thus, we have
\begin{multline*}
  -D(\theta|\theta_0)\leq       \limsup_{n   \to   +\infty}   \frac{1}{n}
\sum_{(n_1\dots,n_k)  \in \N^k}   \pro(Z_n=(n_1,\dots,n_k)) \phantom{sssssssssssssssssssssssssssssssssss}\\ \phantom{sssssssssssss}\times \Big\{\log \frac {\pr\big(\exists m\geq 1, Z_m=(n_1,\dots,n_k)\big)}
{\pro\big(Z_n
=(n_1,\dots,n_k)\big)}+ n_1\sum_x f_0(x)\log \frac{f(x)}{f_0(x)} \Big\}\\
\leq \limsup_{n   \to   +\infty}   \frac{1}{n}  \Big\{\log \sum_{(n_1\dots,n_k)  \in \N^k} \pr\big(\exists m\geq 1,
Z_m=(n_1,\dots,n_k)\big)\phantom{sssssssssssssssssssssssss}\\
\phantom{sssssssssssssssssss}+ \sum_{(n_1\dots,n_k)  \in \N^k} \pro\big(Z_n=(n_1,\dots,n_k)\big) n_1 \sum_x f_0(x)\log \frac{f(x)}{f_0(x)}\Big\}\\
\leq \limsup_{n   \to   +\infty}   \frac{1}{n} \! \Big\{\! \espo[Z^1_n]\!\sum_x \!f_0(x)\log \frac{f(x)}{f_0(x)}\! \Big\}\!=\!
\limsup_{n   \to   +\infty}   \frac{1}{n}   \Big\{n\! \sum_x f_0(x)\log \frac{f(x)}{f_0(x)} \Big\}\!<\!0,
\end{multline*}
as soon as $f\neq f_0$, since $\espo[Z^1_n]=n$ from Assumption 1.

The proof for $D^{*}$ follow the same lines.
\end{proof}\\

It would be interesting to prove the uniqueness of the maximum of
the functions $\ell(\theta)$ and $w(\theta)$ at the true value of
the parameter $\theta_0$. If that was true, the consistency of
maximum likelihood and bayesian estimators would be obtained with
classical arguments (see Arribas-Gil {\em et al.}, 2006). In Section \ref{simus} we
investigate the behavior of functions $\ell(\theta)$ and
$w(\theta)$ via some simulations.
\section{Extension to the case of an arbitrary tree}
Let us now consider an arbitrary phylogenetic tree, that is, a tree with inner nodes such as the one in Figure \ref{arb_tree} (a). Without loss of generality we can assume that we deal with a binary tree (the number of edges going out from every inner node is equal to two) in which the length of the path from the root to each leaf is the same for every leaf in the tree. There is an example of this kind of tree in Figure \ref{arb_tree} (b). Indeed, we will only use this fact to simplify notations, since it allow us to describe the evolutionary behavior of any internal node in a general way and define the model in a simpler manner. Otherwise, the state space of the hidden process would depend on the particular structure of the tree, but the results given in this section still hold.
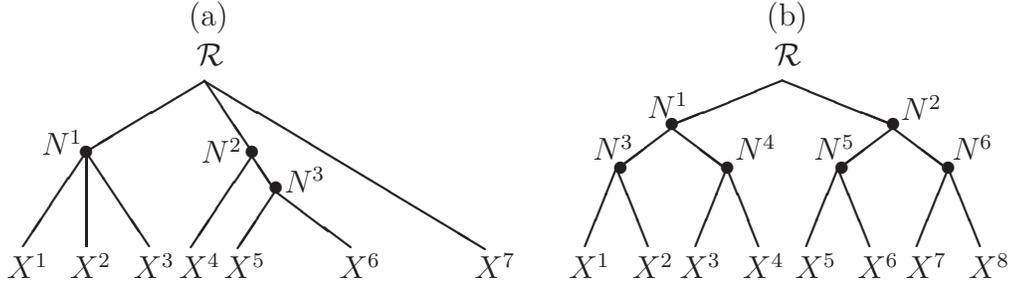
\begin{figure}
\setlength{\unitlength}{1.05mm}
\begin{picture}(60,30)
\thicklines
\put(30,25){\line(-5,-3){15}}
\put(15,16){\line(-2,-3){8}}
\put(15,16){\line(2,-3){8}}
\put(15,16){\line(0,-1){12}}
\put(30,25){\line(5,-3){35.5}}
\put(30,25){\line(2,-3){6}}
\put(36,15.5){\line(-2,-3){7.7}}
\put(36,15.5){\line(2,-3){3}}
\put(39,11){\line(-2,-3){4.8}}
\put(39,11){\line(4,-3){9.5}}

\put(29,27){${\cal R}$}
\put(14,15){$\bullet$}
\put(9.4,15.6){$N^1$}
\put(35,15){$\bullet$}
\put(29.6,14.6){$N^2$}
\put(5,0){$X^1$}
\put(13,0){$X^2$}
\put(21,0){$X^3$}
\put(26.8,0){$X^4$}
\put(38,10.5){$\bullet$}
\put(40.1,10.9){$N^3$}
\put(32.5,0){$X^5$}
\put(47.1,0){$X^6$}
\put(64.2,0){$X^7$}
\put(28,32){(a)}
\end{picture} \hspace{1.2cm}
\begin{picture}(60,30)
\thicklines
\put(30,25){\line(-5,-2){14}}
\put(16,19){\line(-4,-3){7}}
\put(16,19){\line(4,-3){7}}
\put(9,14){\line(-2,-5){4}}
\put(9,14){\line(2,-5){4}}
\put(23,14){\line(-2,-5){4}}
\put(23,14){\line(2,-5){4}}
\put(30,25){\line(5,-2){14}}
\put(44,19){\line(-4,-3){7}}
\put(44,19){\line(4,-3){7}}
\put(51,14){\line(-2,-5){4}}
\put(51,14){\line(2,-5){4}}
\put(37,14){\line(-2,-5){4}}
\put(37,14){\line(2,-5){4}}

\put(29,27){${\cal R}$}
\put(15,18.5){$\bullet$}
\put(12.9,20.3){$N^1$}
\put(8.5,13){$\bullet$}
\put(6,15){$N^3$}
\put(22,13){$\bullet$}
\put(24,15){$N^4$}
\put(3,0){$X^1$}
\put(11,0){$X^2$}
\put(17,0){$X^3$}
\put(25,0){$X^4$}
\put(43,18.5){$\bullet$}
\put(45,20){$N^2$}
\put(50,13){$\bullet$}
\put(51.5,15){$N^6$}
\put(36.5,13){$\bullet$}
\put(34,15){$N^5$}
\put(53.6,0){$X^8$}
\put(45.6,0){$X^7$}
\put(39.5,0){$X^6$}
\put(31.6,0){$X^5$}

\put(28,32){(b)}
\end{picture}
\caption{\emph{(a) An arbitrary phylogenetic tree. (b) A binary phylogenetic tree. ${\cal R}$ stands for the ancestral sequence, $N^i$ stand for sequences in inner nodes (non-observed sequences), and $X^i$ stand for observed sequences.}}\label{arb_tree}
\end{figure}

The multiple-hidden i.i.d. model on a binary tree with $k$ observed sequences, ${\cal T}_{2,k}$, is defined as follows. Consider a sequence of i.i.d. random variables $\{\e_n\}_{n\geq 1}$ on the state space
\begin{multline*}
\E^{{\cal T}_{2,k}}= \left\{ e \in {\cal M}_{(2^k-1),2m}\, ,m \in \N;\,\, e^{h}_{p}\in \{(1,0), (0,0)\}, \phantom{\E^2}\right.\\
\left. e^{\{a_i,i,i'\}}_{p}\in \{(1,0)\!\times \!\E^2\!,\, 0_{3,2}\},\, \mbox{\small $p\!=\!1,\dots,m, \forall h\!\in \!I,\forall i,i'\!\in\!O, i\sim i'$}\right\}
\end{multline*}
where ${\cal M}_{a,b}$ denotes the set of all $a$-by-$b$ natural matrices, $0_{3,2}$ denotes the $3$-by-$2$ null matrix, $I$ denotes the set of internal nodes (unobserved sequences) of the tree and $O$ denotes the set of external nodes (observed sequences) of the tree. For an observed sequence $i$, $a_i$ stands for its direct ancestor, that is, the sequence that is placed in its closest internal node. For two observed sequences $i$ and $i'$, we write $i\sim i'$ if they share the same direct ancestor (that is, $a_i=a_{i'}$). For $e$ in $\E^{{\cal T}_{2,k}}$, $e^{\{i,j,h\}}_{p}$ is the sub-matrix of $e$ composed by rows $i,j,h$ and columns $2(p-1)+1$ and $2p$.

An element $e$ of $\E^{{\cal T}_{2,k}}$ represents the fate of a nucleotide in the root sequence and all the insertions produced at the different levels of the tree. It is a finite sequence of $(2^k-1)$-by-$2$ matrices, in which each row represents one node (sequence) of the tree. We will assume that the first row represents the sequence at the root. The first $(2^k-1)$-by-$2$ matrix represent the fate of the nucleotide at the root in the first column (whether it is conserved, 1, or deleted, 0, in each one of the sequences) and the number of insertions produced to its right in the observed sequences (second column). The difference with the star-tree case, is that now we may also have (non-observed) insertions in the internal sequences. They appear in the following $(2^k-1)$-by-$2$ matrices, represented by 1 in the corresponding position of the first column, where we also represent the fate of the inserted nucleotide and the number of insertions produced to its right in the corresponding descendant sequences. The rows of $e$ that are not concerned by that insertion (because the corresponding sequences are not descendant of that internal sequence) may represent the fate of another inserted nucleotide in a different internal sequence. That is why in the same $(2^k-1)$-by-$2$ matrix we may represent independent events in different rows. Indeed, the events represented in two different rows $i$ and $j$ of the same $(2^k-1)$-by-$2$ matrix are independent if the row corresponding to the closest common ancestor of $i$ and $j$ in that matrix takes the value $(0,0)$, and they are dependent if it takes the value $(1,0)$. There is an example of an element of $\E^{{\cal T}_{2,k}}$ in Figure~\ref{Ex_arb_tree}.
\begin{figure}\begin{center}
{\scriptsize \textbf{Homology structure}  \hspace{5cm} \textbf{\emph{Bare} alignment}\vspace{0.4cm}}\\
{\scriptsize\begin{tabular}{rl|l|l|l|l|l|l|l}
$R$:&        1\,0&0\,0&0\,0&0\,0&0\,0&0\,0&0\,0&0\,0\\
$N^1$:&      1\,0&0\,0&0\,0&1\,0&0\,0&0\,0&1\,0&0\,0\\
$N^3$:&      1\,0&1\,0&0\,0&1\,0&0\,0&0\,0&1\,0&1\,0\\
$X^1$:&      1\,2&1\,0&0\,0&1\,0&0\,0&0\,0&1\,0&1\,0\\
$X^2$:&      1\,0&0\,0&0\,0&1\,1&0\,0&0\,0&1\,2&0\,1\\
$N^4$:&      1\,0&0\,0&0\,0&1\,0&1\,0&1\,0&1\,0&0\,0\\
$X^3$:&      1\,0&0\,0&0\,0&1\,0&1\,0&1\,0&1\,1&0\,0\\
$X^4$:&      1\,1&0\,0&0\,0&1\,1&1\,1&0\,1&0\,0&0\,0\\
$N^2$:&      1\,0&0\,0&0\,0&1\,0&0\,0&0\,0&0\,0&0\,0\\
$N^5$:&      1\,0&1\,0&1\,0&1\,0&0\,0&0\,0&0\,0&0\,0\\
$X^5$:&      1\,0&1\,0&1\,0&1\,0&0\,0&0\,0&0\,0&0\,0\\
$X^6$:&      0\,1&1\,0&0\,1&1\,1&0\,0&0\,0&0\,0&0\,0\\
$N^6$:&      1\,0&1\,0&0\,0&1\,0&1\,0&0\,0&0\,0&0\,0\\
$X^7$:&      1\,0&1\,0&0\,0&0\,0&1\,0&0\,0&0\,0&0\,0\\
$X^8$:&      1\,0&0\,2&0\,0&1\,0&1\,2&0\,0&0\,0&0\,0\\
\end{tabular}}\hspace{0.5cm}
{\scriptsize\begin{tabular}{rl}
$X^1$:&      \texttt{BBB--B------B-----------B---B-}\\
$X^2$:&      \texttt{B-----------BB----------BBB--B}\\
$X^3$:&      \texttt{B-----------B--B-B------B--B--}\\
$X^4$:&      \texttt{B--B--------B-BBB-B-----------}\\
$X^5$:&      \texttt{B-----BB-----------B----------}\\
$X^6$:&      \texttt{----B-B-B----------BB---------}\\
$X^7$:&      \texttt{B--------B-----------B--------}\\
$X^8$:&      \texttt{B---------BB-------B-BBB------}
\end{tabular}}
\end{center}
\caption{An element $e$ of $\E^{{\cal T}_{2,8}}$ (corresponding to the phylogenetic tree in Figure \ref{arb_tree} (b)) and a possible representation of the associated multiple alignment (the choice of the order in which insertions appear is arbitrary). Vertical lines separate the different $15$-by-$2$ submatrices. In this example, $|e|=8$.}\label{Ex_arb_tree}
\end{figure}

For $e\in \E^{{\cal T}_{2,k}}$ such that $e \in {\cal M}_{(2^k-1),2m}$, $m \in \N$, we will note $|e|=m$. Also, for any $(2^k-1)$-by-$2$ submatrix $e_p$, $1\leq p \leq |e|$, we will note $\|e_p\|=e_p(1)+e_p(2)$, that is, the sum of the two columns of $e_p$. $e^{obs}$ will denote the $k$-by-$2|e|$ matrix whose rows are the rows on $e$ corresponding to the observed sequences. For any internal node (non-observed sequence) $i\in I$, $d_i$ will denote the set of the two direct descendants of $i$, and $D_i$ will denote the set of all the descendants of $i$ which are observed sequences. 
Also, for any sequence $i$, and any $p$, $1\leq p \leq |e|$, such that $e_p^{a_i}(1)=1$, we will denote $\overline{\|e_p^i\|}=\sum_{r=p}^q \|e_r^i\|$, where $q$ is such that $e_r^{a_i}(1)=0$ for $r=p+1,\dots,q-1$ and $e_q^{a_i}(1)=1$. $\overline{\|e_p^i\|}$ represents the total number of descendants in sequence $i$ of the given nucleotide from sequence $a_i$. If $i$ is one of the two direct descendants of the root then $\overline{\|e_1^i\|}=\sum_{r=1}^{|e|} \|e_r^i\|$. Note that if $i$ stands for an observed sequence (external node), for any $p$ such that $e_p^{a_i}(1)=1$, $\overline{\|e_p^i\|}=\|e_p^i\|$. The same notations apply to the random process $\{\e_n\}_{n\geq 1}$.

In the case in which we consider the TKF91 indel model, due to the branch independence, the law of $\e_n$, is given by
\begin{equation*}\label{arb_loi_e}
\mathbb{P}_{\lambda}\left(\e_n\!=\!e\right)\!=\!\prod_{p=1}^{|e|}\prod_{\substack{i=2;\\ \mbox{\tiny$e_p^{a_i}\!(1)\!\!=\!\!1$}}}^{2^k-1}
\left(\!q^H_{\substack{\phantom{i}\\ \mbox{\tiny$\overline{\|e_p^i\|}$}}}(t_i)\!\right)^{\!\!\1{e_p^{i}(1)=1}}\!\!
\left(\!q^N_{\substack{\phantom{i}\\ \mbox{\tiny$\overline{\|e_p^i\|}$}}}(t_i)\!\right)^{\!\!\1{e_p^{i}(1)=0}}\!\!, \,
e\in \E^{{\cal T}_{2,k}}, n\geq 1
\end{equation*}\\
where $t_i$ represents the evolutionary time between sequences $i$ and $a_i$. In the general case we will note $\pi$ the law of $\e_n$.

As in the star tree case, the process $\{\e_n\}_{n\geq 1}$ generates a random walk
$\{Z_n\}_{n \geq 0}$ with values on $\N^k$ by letting $Z_0
=0_k$ and $Z_n =\sum_{1\leq j\leq n} \sum_{1\leq p\leq |\e_j|}\|\e_{j_p}^{obs} \|$ for
$n\geq 1$. The coordinate random variables corresponding to $Z_n$
at position $n$ are denoted by $(Z^1_n,\dots,Z^k_n)$ ({\it i.e.}
$Z_n=(Z^1_n,\dots,Z^k_n)$).

Let us now describe the emission of the observed sequences which take values on a finite alphabet ${\cal A}$. We distinguish to kinds of emissions, joint emissions across $k$ or a smaller number of sequences (corresponding to $\e_{n_p}(1)$, $1\leq p \leq |\e_n|$) and single emissions (corresponding to $\e_{n_p}(2)$, $1\leq p \leq |\e_n|$).
For $n\geq1$, and for $1\leq p \leq |\e_n|$, if  $\e_{n_p}^{\tau}(1)=1$ and $\e_{n_p}^{a_\tau}(1)=0$ for any $\tau \in I$, then a vector of
$r=|\{i\in D_{\tau}|\e_{n_p}^{i}(1)=1\}|$ r.v. is emitted according to some probability
distribution $h_J$, $J=\{i\in D_{\tau}|\e_{n_p}^{i}(1)=1\}$, on ${\cal A}^r$ and
$\sum_{i\in D_{\tau}} \e_{n_p}^{i}(2)$ r.v. $\{ X^i_{1:\e_{n_p}^{i}(2)}\}$,
$i\in D_{\tau}$, are emitted according to the following
scheme: $\{X^i_j\}^{i\in D_{\tau}}_{1,\e_{n_p}^{i}(2)}$ are independent and
identically distributed from some probability distribution $f$ on
${\cal A}$.
\begin{remark}
In practice, the emission law $h$, may take into account the emissions in internal sequences. Consider, for instance, the emission in the first column of the homology structure of Figure \ref{Ex_arb_tree}. If we deal with a classical markovian  substitution model, with stationary distribution $\nu$ and  transition probability matrix $p_t(\cdot,\cdot)$, the emission of nuleotides $x^1,\dots,x^5,x^7,x^8$ in sequences $X^1,\dots,X^5,X^7,X^8$ would have probability
\begin{multline*}
h_{\{1,\dots,5,7,8\}}(x^1,\dots,x^5,x^7,x^8) \\
=\sum_{R \in \A} \nu(R)\times \left\{ \left( \sum_{\tau_1 \in A} p_{{s_1}}
(R,\tau_1) \left[\sum_{\tau_3 \in A} p_{{s_3}}
(\tau_1,\tau_3)  p_{t_1}(\tau_3,x^1) p_{t_2}(\tau_3,x^2) \right] \right.\right. \hspace{3cm}\\ \phantom{a}\hspace{6cm} \left.
\times \left[\sum_{\tau_4 \in A} p_{{s_4}}
(\tau_1,\tau_4)  p_{t_3}(\tau_4,x^3) p_{t_4}(\tau_4,x^4) \right] \right)\\ 
\times \left( \sum_{\tau_2 \in A} p_{{s_2}}
(R,\tau_2) \left[\sum_{\tau_5 \in A} p_{{s_5}}
(\tau_2,\tau_5)  p_{t_5}(\tau_5,x^5) \right] \right.\hspace{6cm}\\ \phantom{a}\hspace{5cm} \left. \left.
\times \left[\sum_{\tau_6 \in A} p_{{s_6}}
(\tau_2,\tau_6)  p_{t_7}(\tau_6,x^7) p_{t_8}(\tau_6,x^8) \right] \right) \right\}
\end{multline*}
where $R$ represents the nucleotide in the root, $\tau_i$ the nucleotide in internal sequence $N^i$, $s_i$ the evolution time to internal sequence $N^i$ from its direct ancestor and $t_i$ the evolution time to observed sequence $X^i$  from its direct ancestor.  

\end{remark}
As in the star tree case, conditionally to the process $\{\e_n\}_{n\geq 1}$, the
random variables emitted at different instants are independent. The whole multiple-hidden i.i.d. model is described by the parameter $\theta=(\pi,\,\{h_J\}_{J\subseteq K},\,f)\in \Theta$.

The conditional distribution of the observations given an homology structure $e_{1:n}=(e_j)_
{1\leq j\leq n}$, writes
\begin{eqnarray}\label{arbicondi}
&&\pr(\X_{1_k:Z_n} |  \e_{1:n}=e_{1:n} ) = \prod_{j=1}^n  \pr(\X_{Z_{j-1}+1_k:Z_j} |  \e_j=e_j ) \nonumber\\
&=&\prod_{j=1}^n \prod_{p=1}^{|e_j|} \Big\{ \prod_{\substack{\tau\in I;\\ \mbox{\tiny$e_{j_p}^{\tau}\!(1)\!\!=\!\!1$}\\ \mbox{\tiny $e_{j_p}^{a_{\tau}}\!(1)\!\!=\!\!0$}}}   h_{\{i\in D_{\tau} | e_{j_p}^{i}(1)=1\}}
\left(\{X^{i}_{Z^{i}_{j-1}+\sum_{r=1}^{p-1}\|e_{j_p}^{i}\|+1}\}_{\{i\in D_{\tau} | e_{j_p}^{i}(1)=1\}}\right) \Big\}  \nonumber\\
&& \hspace{1cm}\times \Big\{ \prod_{i\in O} \prod_{s=1}^{e_{j_p}^i(2)}
 f\big(X^i_{Z^i_{j-1}+\sum_{r=1}^{p-1}\|e_{j_p}^i\|+e_{j_p}^i(1) +s}\big) \Big\}.
\end{eqnarray}

And the complete distribution $\pr$ is given by
\begin{eqnarray*}
&&\pr(\e_{1:n}=e_{1:n},\X_{1_k:Z_n})=  \pr(\X_{1_k:Z_n} | \e_{1:n}=e_{1:n} ) \prod_{j=1}^n
\pi(e_j).
\end{eqnarray*}

At this point we can define the parameter set $\Theta_0$, likelihoods $\omega_n (\theta)$ and $\ell_n(\theta)$ and divergence rates $D (\theta \vert \theta_0)$ and $D^{*}(\theta \vert \theta_0)$ in the same way as in the star-tree case. Indeed Theorem 1 also holds in this case. Moreover, since we do not exploit any specific characteristic of $\pi$ or the emission laws to prove this result, the proof is exactly the same as the one given in Section \ref{Inf_Div}. The only slightly difference appears when proving point 3, but it is clear that $\pr(\X_{1_k:Z_1})>0$ also holds in this case for $\theta \in \Theta_0$.

By analogy to the star tree case, we will establish an assumption to ensure that asymptotic results for $n \to \infty$
will imply equivalent ones for $n_i \to \infty,\, i=1,\dots,k$. It also guarantees that $\esp[Z_n]=n$, for $n\in \N$, as it is required to prove Theorem 2.
\begin{assumption} In the multiple-hidden \vspace{-0.3cm}i.i.d. model on a binary tree $\esp\,\big[\displaystyle\sum_{p=1}^{|\e_n|}\|\e_{n_p}^{obs}\|\big]=1_k$, for $n\geq 1$, for any $\theta\in\Theta$.
\end{assumption}
This assumption holds for the multiple-hidden i.i.d. model under the TKF91 indel evolution process as it is shown in the following lemma.
\begin{lemma}\label{esp=1arbitrary}
In the multiple-hidden i.i.d. model on a binary tree under the TKF91 indel evolution process, for any $\lambda >0$ we have $Z_n^i \sim n$, $i=1,\dots,k$, $\mathbb{P}_{\lambda}$-almost surely.
\end{lemma}
\begin{proof}.
We have already proved this result in the case of a star phylogenetic tree (Lemma \ref{esp=1}), that is, when we have a tree without internal nodes. Now, the idea of the proof, is that, if at each level of the tree the expectation of the number of nucleotides descending (conserved plus inserted) from a single nucleotide in the parent sequence is 1, the expectation of the total number of nucleotides at each observed sequence descending from a single nucleotide in the root sequence will also be 1. Let us show it recursively. 

Let $L$ be the total number of levels on the tree, that is the number of edges between the root and an observed sequence (in the case of a binary tree, $L=\ln_2 k$). For each observed sequence $i\in Obs$, we will note $a_i^l$, $l=1\dots,L$ the $l$-th ancestor of $i$, beginning at the direct ancestor and ending at the root of the tree. For all $i\in Obs$ and for all $n\geq 1$ we have that
$$Z_n^i=\sum_{1\leq j\leq n} \sum_{1\leq p\leq |\e_j|}\|\e_{j_p}^i \|$$
where $\{\e^i_j\}_{j\geq 1}$ are i.i.d. Moreover, we have, for any $\lambda >0$
\begin{eqnarray*}
& &\!\!\!\!\!\!\!\!\!\!\!\mathbb{E}_{\lambda}\!\left[\sum_{p=1}^{|\e_j|}\|\e_{j_p}^i \|\right]
\stackrel{(a)}{=}\mathbb{E}_{\lambda}\!\!\left[\sum_{p=1}^{\sum_{q=1}^{|\e_j|}\|\e^{a_i^1}_{j_q}\|}\!\!\!\|\e_{j_p}^i \|\right]=\mathbb{E}_{\lambda}\left\{ \mathbb{E}_{\lambda}\left[\sum_{p=1}^{\sum_{q=1}^{|\e_j|} \|\e^{a_i^1}_{j_q}\|}\!\!\!\|\e_{j_p}^i \| \,\left| \sum_{q=1}^{|\e_j|} \|\e^{a_i^1}_{j_q}\|\right.\right]\!\right\}\\
&=&\mathbb{E}_{\lambda}\left[ \sum_{p=1}^{\sum_{q=1}^{|\e_j|} \|\e^{a_i^1}_{j_q}\|} \mathbb{E}_{\lambda}\left[\|\e_{j_p}^i \| \right]\right]\stackrel{(b)}{=} \mathbb{E}_{\lambda}\left[ \sum_{p=1}^{\sum_{q=1}^{|\e_j|} \|\e^{a_i^1}_{j_q}\|} 1\right]=\mathbb{E}_{\lambda}\left[ \sum_{q=1}^{|\e_j|} \|\e^{a_i^1}_{j_q}\|\right]\\
&=&\mathbb{E}_{\lambda}\left[ \sum_{q=1}^{|\e_j|} \|\e^{a_i^2}_{j_q}\|\right]=\dots=\mathbb{E}_{\lambda}\left[ \sum_{q=1}^{|\e_j|} \|\e^{a_i^{L-1}}_{j_q}\|\right]\stackrel{(c)}{=}1
\end{eqnarray*}
where (a) comes from the fact that $\|e_{j_p}^{i}\|\neq 0$ only for those $p$ such that $\e_{j_p}^{a^1_i}(1)=1$, and (b) comes from Lemma \ref{esp=1}. Finally, for any $i \in Obs$, $\sum_{q=1}^{|\e_j|} \|\e^{a_i^{L-1}}_{j_q}\|$ is just the number of descendants (conserved plus inserted nucleotides) of the nucleotide in the root in one of its direct children. The expectation of this quantity is again 1 by Lemma~\ref{esp=1}. The result holds from the strong law of large numbers.
\end{proof}\\

Finally, to prove that Theorem 2 also holds in the case in which we deal with an arbitrary tree, we need to show that for any $\theta \in \Theta_{marg}$ (same definition as in Section~\ref{Inf_Div}) and for any observed sequence $i$
$$ \pr\big(Z_n=(n_1,\dots,n_k),\,X^i_{1:n_i}=x^i_{1:n_i}\big)=\pr\big(Z_n=(n_1,\dots,n_k)\big)f^{\otimes n_i} (x^i_{1:n_i}).$$
But this can be easily shown from expression (\ref{arbicondi}) in the same way that in (\ref{re}).

Then the asymptotic results obtained in Section \ref{Inf_Div} are also valid when the phylogenetic tree has a general form. 

\section{Simulations}\label{simus}
We have considered for the simulations a 3-star phylogenetic tree,
the most simple non trivial example of multiple alignment.
The branches lengths, or evolutionary distance from the ancestral sequence to the observed sequences, are
set
to $1$ in all branches. Let us recall that this distance is not the real time of evolution between sequences but a measure
given in terms of the number of expected
evolutionary events per site. Indeed, under the TKF91 indel evolution model $\lambda t$ is the expected number of indels per site between two sequences at distance $t$.

The distribution of the hidden process has been taken to be the distribution of the homology structure under the TKF91 indel evolution model, that is, $\{\e_n\}_{n\geq 1}$ are independent and identically distributed as in (\ref{loi_e}). However, we have used the equivalent multiple-HMM (see for instance Hein {\em et al.}, 2003, and Figure \ref{mapa})
scheme to simulate the sequences. Indeed, in practice it is easier
to simulate from a finite state Markov chain than from our i.i.d.
variables on $\N^3$. The number of states for the Markov chain for
three sequences is 15 ($2^4-1$). The simulated sequences have been
used to compute the quantities $\ell(\theta)$ and $w(\theta)$. The
log-likelihood $\omega_n(\theta)$ has been computed with the
Forward algorithm for multiple-HMM (cf. Durbin {\em et al.}, 1998). Note that
this algorithm computes the log-likelihood by summing over all
possible alignments of the three sequences. However, since a
homology structure is just a set of alignments, this is equivalent
to sum over all possible homology structures, and the final
result is exactly $\omega_n(\theta)$. The time complexity for a
non-improved version of this algorithm is $O(15^2 n_1 n_2 n_3)$,
where $n_1$, $n_2$ and $n_3$ are the lengths of the observed
sequences. Computation of $\ell_n(\theta)$ is done with a modified
version of the Forward algorithm that takes into account the
length of the ancestral sequence. The time complexity grows now to
$O(15\, n\, n_1 n_2 n_3)$. This is the reason for having limited
the simulations to 3 sequences.

The emission distributions chosen for the simulations, $\{h_{J}\}_{J\subseteq \{1,2,3\}}$ and $f$, are defined by the substitution model described below.
\subsection{The substitution model}
For the whole simulation procedure we consider the following
pairwise markovian substitution model:
\begin{equation*}
p_t(x,y) =\left\{
  \begin{array}{ll}
 (1-e^{-\alpha t})\nu(y) & \text{ if } x\neq y\\
\{ (1-e^{-\alpha t})\nu(x) +e^{-\alpha t} \} & \text{
otherwise, }
  \end{array}
\right .
\end{equation*}
where $\alpha  >0$ is called the substitution rate, $t$ is the
evolutionary distance, and for every letter $x$, $\nu(x)$ equals the
equilibrium probability of $x$. This model is known as the Felsenstein81 substitution model (Felsenstein, 1981).
We will take $f(\cdot)=\nu(\cdot)$. We define the
emission function $h$ as
$$h_{J}((x_i)_{i\in J})=\sum_{R\in \A} \nu(R)\prod_{i \in J}p_{t_{i}}
(R,x_i)$$ for all $J \subseteq \{1,2,3\}$.

The equilibrium probability
distribution $\nu(\cdot)$ is assumed to be known and will not be  part of the
parameter. Then we have $f(\cdot)=f_0(\cdot)$. We will set it to $\{\frac{1}{4},\frac{1}{4},\frac{1}{4},\frac{1}{4}\}$
for the whole simulation procedure.
The unknown parameter is $\theta=(\lambda,\alpha)$.
\subsection{Simulation results}
We have computed the functions $\ell(\theta)$ and $w(\theta)$ for
two different values of $\theta_0$:
\begin{itemize}
\item $\lambda_0=0.02$, $\alpha_0=0.1$ and \item $\lambda_0=0.01$,
$\alpha_0=0.08$.
\end{itemize}
The substitution rate is much bigger than the insertion-deletion
rate and both are quite small, as expected by biologists.\\
\begin{figure}
\centering
\includegraphics[width=6.95cm]{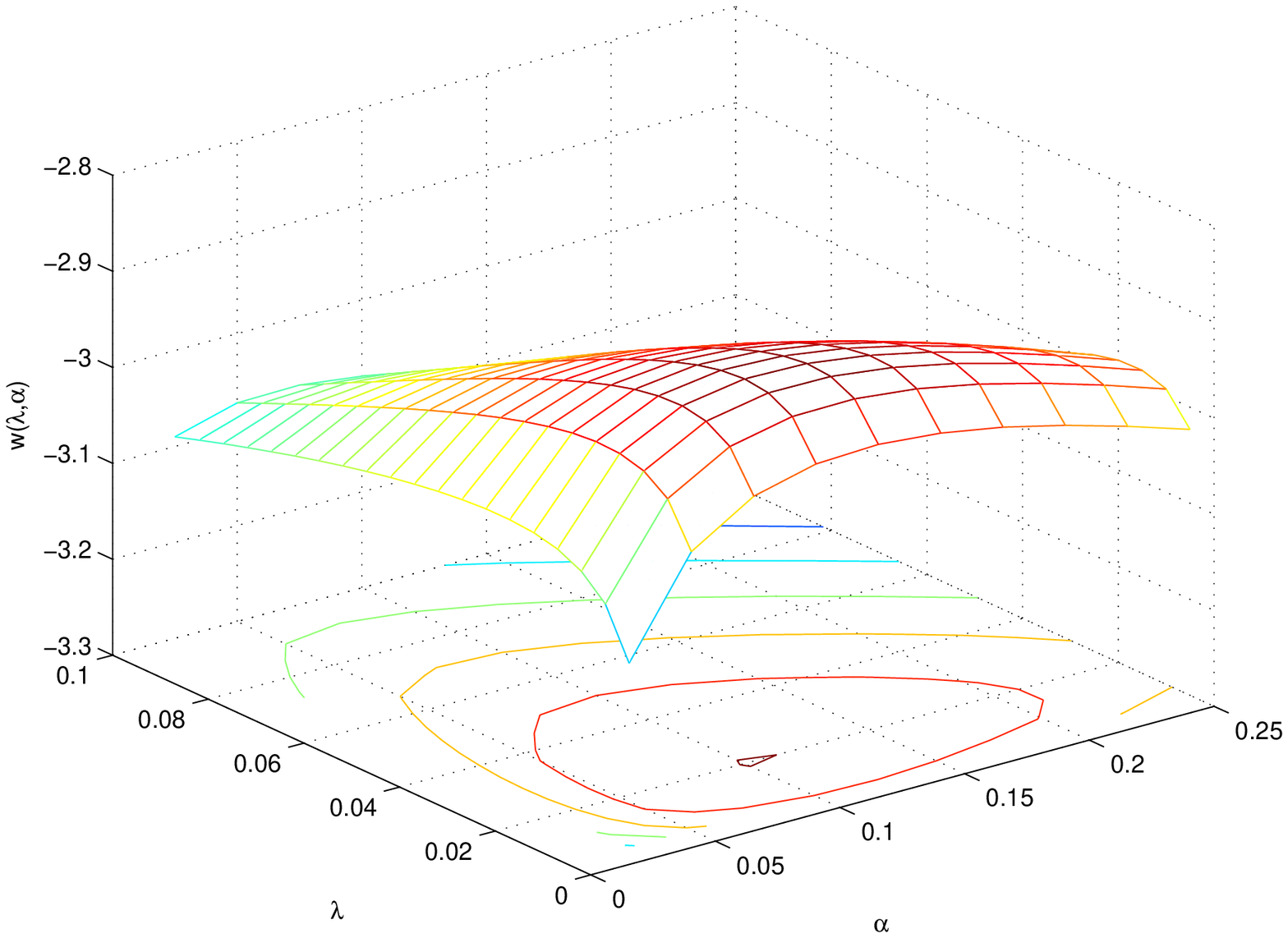}\hspace{-0.2cm}\includegraphics[width=6.95cm]{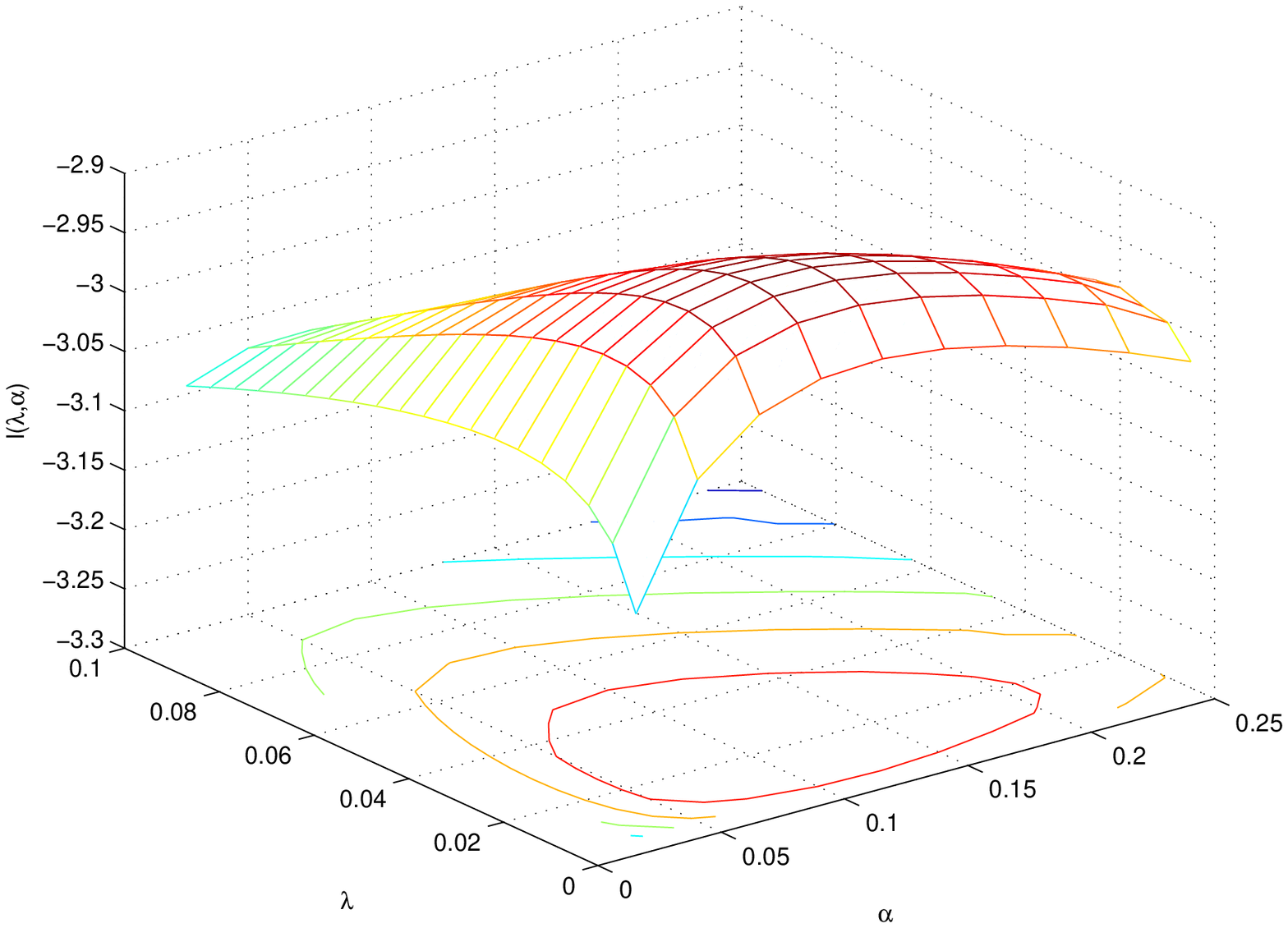}\vspace{1cm}\\
\includegraphics[width=13.9cm]{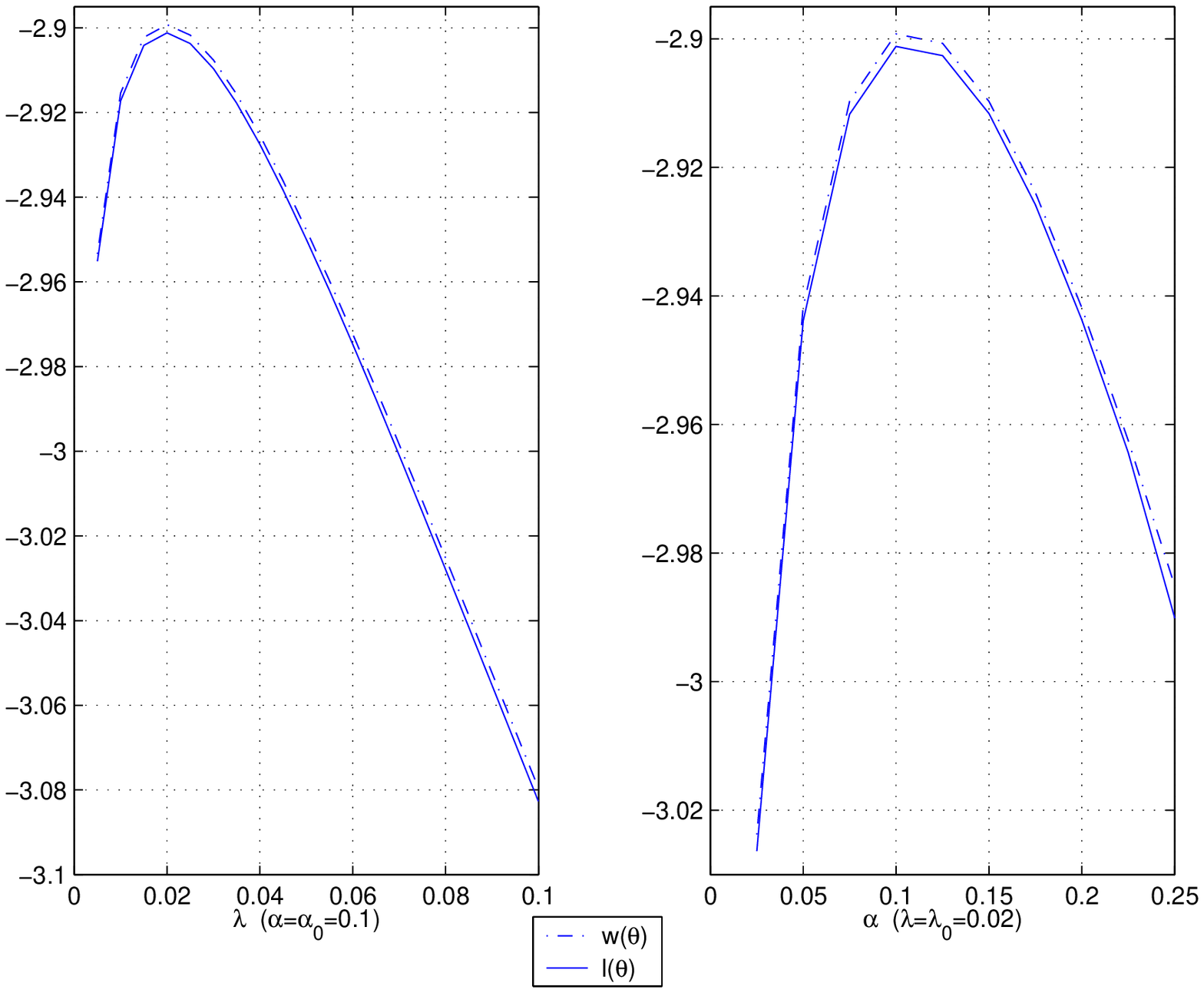}
 \caption{\emph{On top: $w$ and $\ell$ for parametrization
 ($\lambda_0=0.02,\alpha_0=0.1$). On bottom: cuts of $\ell$ and $w$ for
 $\alpha=\alpha_0$ fixed and for $\lambda=\lambda_0$ fixed.}}
\label{w}
\end{figure}
\begin{figure}
\centering
\includegraphics[width=6.95cm]{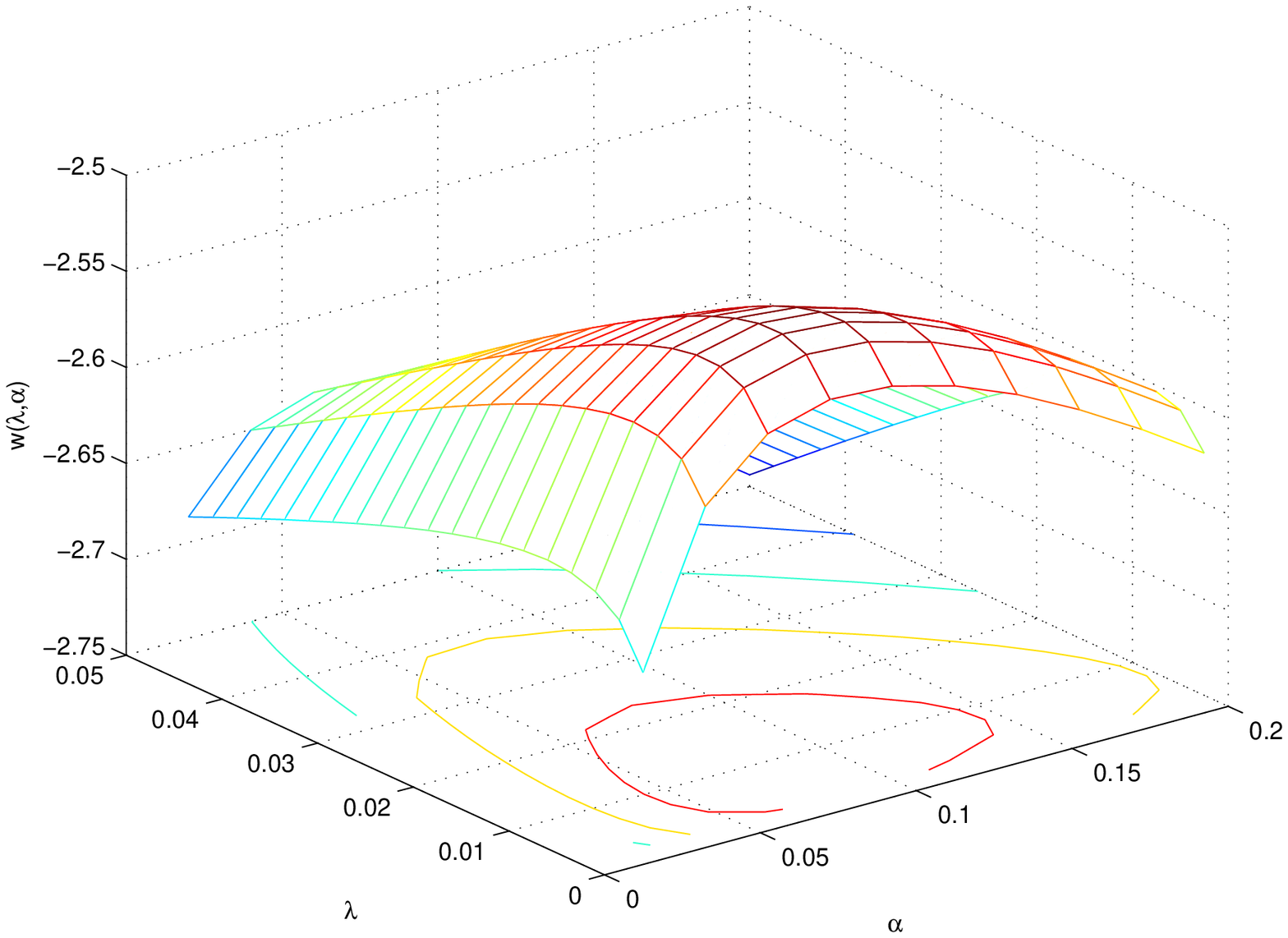}\hspace{-0.2cm}\includegraphics[width=6.95cm]{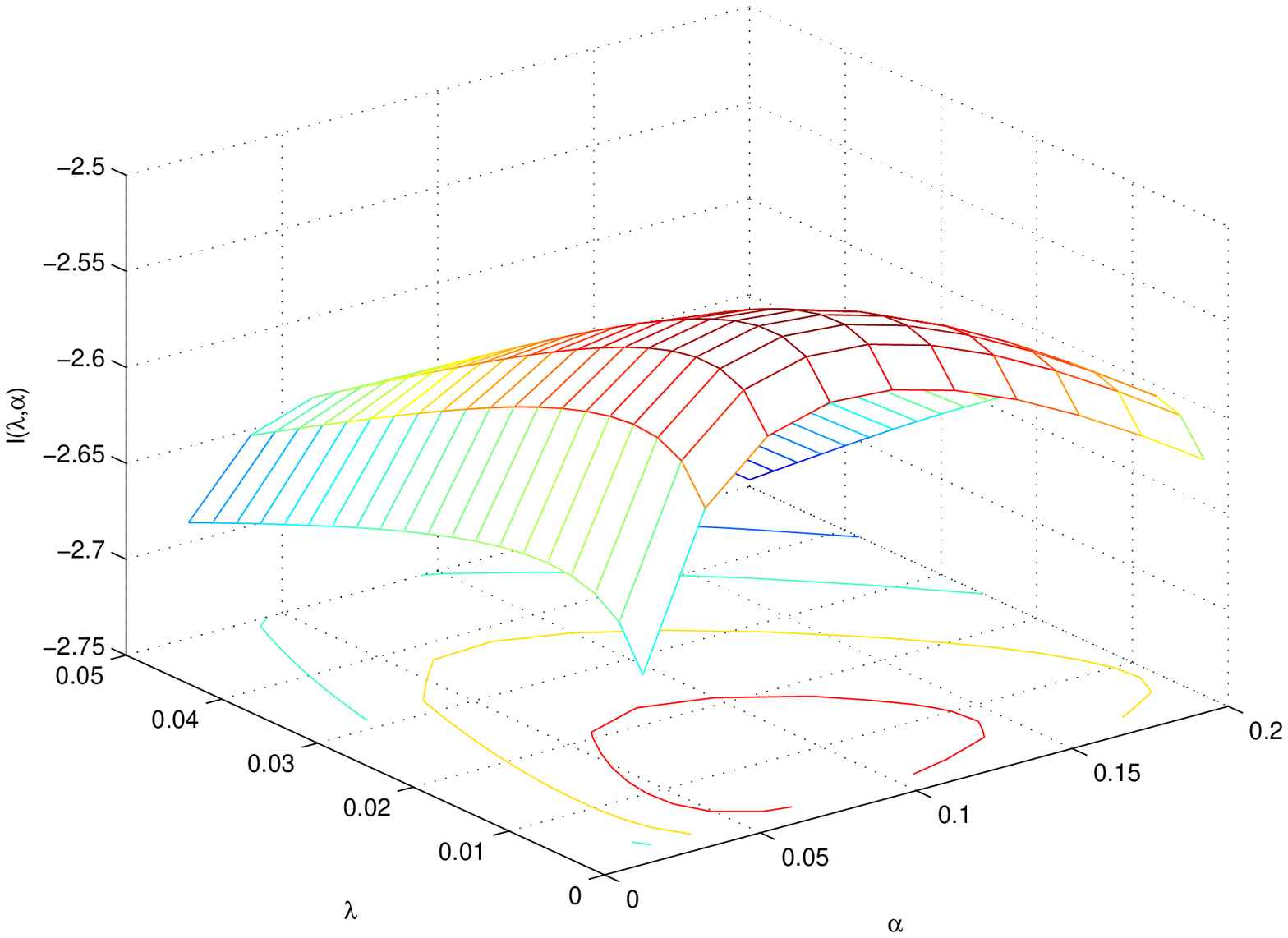}\vspace{1cm}\\
\includegraphics[width=13.9cm]{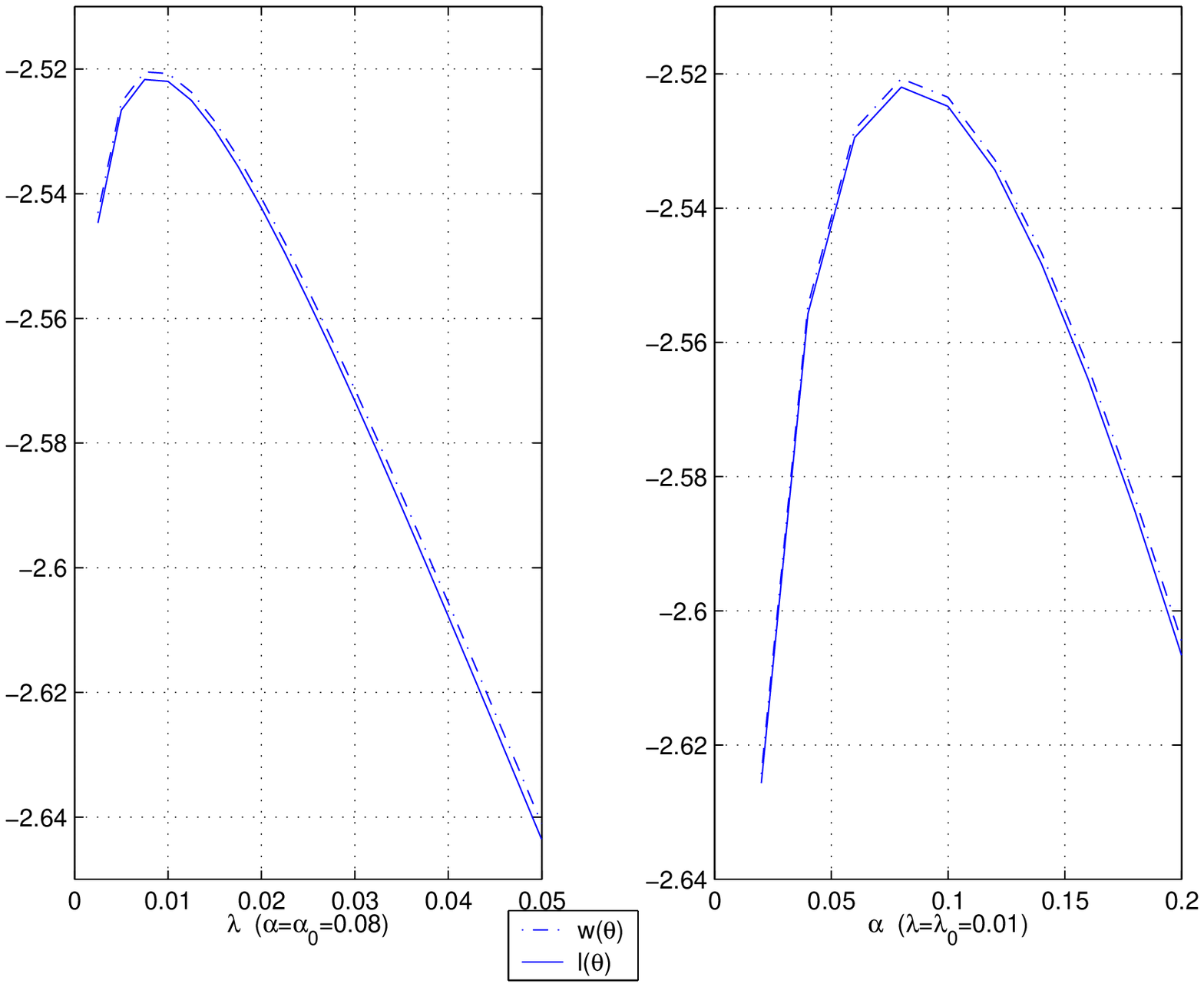}
 \caption{\emph{On top: $w$ and $\ell$ for parametrization
 ($\lambda_0=0.01,\alpha_0=0.08$). On bottom: cuts of $\ell$ and $w$ for
 $\alpha=\alpha_0$ fixed and for $\lambda=\lambda_0$ fixed.}}
\label{w2}
\end{figure}
The graphs of $\ell(\theta)$ and
$w(\theta)$ for these parameterizations are shown in Figures
\ref{w} and \ref{w2}. For the first parametrization we can see that $w(\theta)$ seems to take its
maximum at ($\lambda_0,\alpha_0$) (Figure \ref{w}, top left). For $\ell(\theta)$ this is not so evident.
Neither for any of the two functions for the second parametrization. However, when looking at the cuts
of $w(\theta)$ and $\ell(\theta)$ for $\alpha=\alpha_0$ and $\lambda=\lambda_0$ we appreciate that in
both parameterizations both seem to take their maximums near $\lambda_0$ and $\alpha_0$ respectively.
We remark that in the two examples, the functions $\ell(\theta)$ and
$w(\theta)$ are very close to each other.

\section{Discussion}
The main contribution of this work is to provide a probabilistic and statistical
background to parameter estimation in the multiple alignment of sequences based on a rigorous
model of evolution. We describe the homology structure of $k$ sequences related by a star-shaped phylogenetic
tree as a sequence of i.i.d. random variables whose distribution is determined by the evolution process. Given the observed sequences, the homology structure is a latent (non-observable) process. We formally define the latent variable model that {\em emits} the observed sequences, namely the multiple-hidden i.i.d. model. We discuss possible definitions of likelihoods in comparison with the quantities computed by multiple alignment algorithms. Our main results are given in Theorems~\ref{thdivergence2} and \ref{contrast2}, where we first prove the convergence of normalized log-likelihoods and identify cases where a divergence property holds. We then extend the definition of the model and the results obtained to the case of an arbitrary phylogenetic tree.

Despite the positive results that we obtain, it is not yet possible to validate the estimation of evolution parameters under the multiple-hidden i.i.d. model in every situation. However, the simulation studies that we present to investigate
situations that are not covered by Theorem \ref{contrast2} provide encouraging results.

\section*{Acknowledgments}
The author would like to thank Elisabeth Gassiat from Université Paris-Sud (France) and Catherine Matias from Génopole, CNRS (France), for fruitful advice and helpful comments. The author was partially supported by the Spanish Ministerio de Ciencia e Innovación through project ECO2008-05080 and by Comunidad de Madrid - Universidad Carlos III (Spain) through project CCG08-UC3M/HUM-4467.
\bibliographystyle{miestilo2}
\bibliography{biblio_mult_align}
\end{document}